\documentclass[11pt]{article}
\usepackage{fullpage,url,amsmath,amsfonts,amssymb, tedmath}
\newcommand{\D}{\mathcal{D}} 
\newcommand{\ord}{$\rm{OrderMod}$}

\title{Linear block and convolutional MDS 
  codes to required rate, distance and type
 }% attainable}
\author{Ted
 Hurley\footnote{National Universiy of Ireland Galway, email:
 Ted.Hurley@NuiGalway.ie}}
\date{} 
\begin{document}
\maketitle

\begin{abstract}\let\thefootnote\relax\footnote{\noindent Keywords:
    code, MDS, dual-containing, QECC, LCD, convolutional

MSC Classification: 94B05, 11T71, 16S99}
 Algebraic methods for the 
 design  of series of maximum distance separable (MDS) linear block
 and convolutional codes to required specifications and types are
 presented. Algorithms are given to design codes to required rate and
 required 
 error-correcting capability and required types.  Infinite
 series of block  codes 
 with rate approaching a given rational $R$ with $0<R<1$ and relative
 distance over length approaching $(1-R)$ are designed. These can be
 designed over fields of given characteristic $p$ or over fields of
 prime order and can be specified to be of a particular type  
 such as (i) dual-containing under Euclidean inner product, (ii)
 dual-containing under Hermitian inner product, (iii)
 quantum error-correcting, (iv) linear complementary dual (LCD).
 %% Methods for designing both linear block  and convolutional such codes are
 %% devised. The codes may be required to be over fields of
 %% characteristic a given prime or over prime fields.
 Convolutional 
 codes to required rate and distance and
 infinite series of convolutional codes
 with rate approaching a
 given rational $R$ and distance over length approaching $2(1-R)$ are designed. 
%% The general methods lead to the construction  and
%% analysis of series and infinite series of such codes dual containing
%% convolutional  codes leading to quantum convolutional codes.
The designs are algebraic and properties, including distances, are shown  
 algebraically.  Algebraic explicit  efficient decoding methods are referenced. % for some of the
 %constructions.  
 %% Convolutional codes at or near the maximum free distances attainable
%% are  constructible. %  be of a special required  type.  %% convolutional codes. 

%% The methods enable the construction, design and analysis of  series of
%% LDPC (low
%% density parity check) convolutional codes including decoding schemes.  
%% Self-dual and dual-containing convolutional codes may also be designed by the methods; dual-containing codes enables the construction of quantum codes.  
\end{abstract}
\section{Introduction}
\subsection{Motivation, summary}
Linear block and convolutional codes are error-correcting codes which 
are used extensively in many  applications 
including digital video,
radio, mobile communication, and satellite/space
communications. {\em  Maximum distance separable} (MDS) codes are of
particular interest and  {\em dual-containing codes}, which lead
to the design of {\em quantum error-correcting codes}, QECCs, and {\em linear complementary
dual}, LCD,  codes are also of great interest with many  application. %% Of
%% particular interest are low density parity check (LDPC)
%% convolutional
%% codes and dual containing convolutional codes which lead to the design
%% of quantum convolutional codes.
Codes for which there exist efficient decoding methods are required
and necessary for applications.

This paper gives  design methods for  both linear block codes and
convolutional codes of the highest distance possible for a particular
length and rate. The design methods are then extended to particular {\em types} of codes. {\em Types} here include DC (dual-containing), 
  QECC (quantum error-correcting codes), LCD (linear complementary dual). {\em MDS convolutional codes} are designed where MDS here means the codes attain the GSB (generalised Singleton bound, see section \ref{convo} below for definition) for convolutional codes.
The methods allow the design of codes to given 
specifications and to design infinite series of such codes. The block
linear MDS DC codes are  designed under both 
 Euclidean and  Hermitian inner products and lead to
MDS %% \footnote{MDS may have 
  %% different parameters according to the types of code under consideration}
QECCs under Euclidean or Hermitian  inner products. %%  LCD codes
%% and then stretching these to convolutional MDS, MDS dual-containing,
%% MDS quantum and MDS LCD codes.

%% After some basic constructions, a number of algorithms are
%% given. These are designed to construct codes and MDS codes to given
%% rate and given error-correcting capability and then to specify these
%% for given types of codes. 
For given (allowable) rate $R$ and given distance $(2t+1)$ (that is, for
specified error-correcting capability), 
design methods for MDS block linear codes  with efficient
decoding algorithms are given. %% to
%% specific rate $R$ and to specific error-correcting capability with specified
%% distance $2t+1$ are given together with efficient decoding algorithms.
Design methods for {\em types} of such are then derived, where
`types' can be DC  
or LCD; QECCs are obtainable from DC codes. These  are 
further specified to be over fields of
characteristic a fixed prime $p$ or over  a field of prime
order. In fields of prime order the arithmetic is modular arithmetic which is particularly nice and efficient.
Infinite series of codes, which can be required to be DC or LCD, in
which the rate approaches a given $R$ and the relative
distance (ratio of distance over length) approaches $(1-R)$ are
designed. These infinite series can  be  specified  to be
codes with characteristic a given prime $p$ or to be  
codes over a field of prime order. Again note that QECCs are
obtainable from DC codes, and if the DC code is MDS then the QECC 
obtained is maximum distance attainable for a QECC with those parameters. 

In general the convolutional codes designed offer 
better distances than the equivalent block linear code of the same
length and rate. Note that the convolutional codes   are {\em designed
  algebraically}.\footnote{There exist very few algebraic
  constructions for designing convolutional 
codes and  search methods limit their size and
availability, see McEliece \cite{mac} for discussion and
also \cite{alm1,alm2,guar,mun}.}

%% These codes
%% are often implemented in concatenation with a hard-decision code,
%% particularly Reed Solomon. Prior to turbo codes, such constructions
%% were the most efficient, coming closest to the Shannon limit.
%The literature on convolutional codes is extensive and expanding. 
  \subsection{Background, notation} Background on  coding theory  
  may be found in \cite{blahut}, 
  \cite{joh}, \cite{mceliece},  \cite{mac},  \cite{macslo} 
  and many others.

 %% separabl (MDS) code is one where the Singleton bound is attained,
 %% that is, one of the form $[n,r,n-r+1]$. 
%Most of the codes available in the literature are constructed
%case by case and by computer search 
%% \cite{mac} notes that 
%% many of the existing convolutional
%% codes
%% had been found by  pointing out  % for  information on and analysis of 
% convolutional codes 
%% There exists very few algebraic constructions for designing convolutional
%% codes and  search methods limit their size and
%% availability, see McEliece, \cite{mac} for discussion and
%% \cite{alm1,alm2,guar,mun} for some algebraic constructions. %the lack of algebraic methods for their  construction.    
%% Here algebraic formulations using unit schemes are developed 
%% for designing and analysing convolutional codes. % There is no limit on
% the length.  

  The notation for linear block codes is fairly standard.   Here
$[n,r,d]$ denotes a linear block  code of 
  length $n$, dimension $r$ and distance $d$. %% A maximum distance
The maximum distance attainable by an $[n,r]$
linear block code is $(n-r+1)$ and this is known as the {\em Singleton bound}, 
 see \cite{blahut} or \cite{mac}. A linear code
$[n,r]$ attaining the maximum distance possible  
is called an {\em MDS} (maximum
distance separable) code. The ratio of
distance over length features here and we refer to this as the
{\em relative distance}, rdist for short, of the code.  
The MDS linear block codes are those with
maximum error correcting capability for a given length and dimension. MacWilliams and Sloane refer to
MDS codes in their book \cite{macslo} as ``one of the most fascinating
chapters in all of coding theory''; MDS codes are equivalent to geometric objects called n-arcs and combinatorial objects called orthogonal arrays, \cite{macslo},
and are, quote, ``at the heart of combinatorics and finite
geometries''.

%  \paragraph{Why MDS?} 
 A dual-containing, DC,  code $\C$ is a code which contains its dual
 $\C^\perp$; thus a DC code is a code $\C$ such that $\C\cap \C^\perp =
 \C^\perp$. % where $\C^\perp$ denotes the dual code of $\C$.

 A  linear complementary dual, LCD,  code is one such that its
intersection with its dual is zero, that is, it's a code $\C$ such that
$\C\cap \C^\perp = 0$ 

LCD codes and DC codes are `supplemental' to
one another in the sense that $\C$ is DC if
$\C\cap\C^\perp = \C^\perp$ and $\C$ is LCD if $\C \cap \C^\perp = 0$.
We shall see this
further in action when MDS DC block linear codes are extended to MDS LCD
convolutional codes and LCD MDS codes are extended to MDS DC 
convolutional codes. %% where the distance of the convolutional codes are
%% of the order of twice the distance of the
%% corresponding MDS linear code with the same length and rate.

{\bf Why DC?} DC codes have been studied extensively in particular
 since they lead by the CSS construction to the design of quantum
 error-correcting codes, QECCs, 
%% QECCs are designed  from DC codes by CSS method,
see \cite{calderbank} and also \cite{quantum}.  
%% DC MDS codes  lead to the design of MDS QECCs. DC MDS
%% convolutional codes 
%% lead to  the design of quantum MDS convolutional codes.
The CSS constructions are specified  as follows: 
\begin{itemize} 
\item Let $\mathcal{C}$ be a linear block code $[n,k,d]$ over $GF(q) $ containing its dual $\mathcal{C}^\perp$. The CSS construction derives a quantum (stabilizer)  $[[n,2k-n,\geq d]]$ code over $GF(q)$. 
\item Let $\mathcal{D}$ be a linear block code over $GF(q^2)$ containing its Hermitian dual  $\mathcal{D}^{\perp_H}$. The CSS construction derives a quantum (stabilizer) code $[[n,2k-n, \geq d]]$ code over $GF(q^2)$.
\end{itemize}
For more details on CSS constructions of QECCs 
see \cite{ash,ketner}; proofs of the above  may also be found therein. The work of \cite{ash} follows from Rains' work on non-binary codes \cite{rains}.  
%% in \cite{ash} and \cite{ketner}.  
As noted in for example  \cite{ash} if the DC code used for the CSS construction is an MDS linear code then the quantum code obtained is a {\em quantum MDS code} which means it has the best possible distance attainable for such a quantum code.

{\bf Why LCD?} LCD codes have been studied extensively in the literature. For
background, history and general theory %and more constructions %of
on LCD codes,
consult the nice articles \cite{sihem3,sihem2,sihem4, sihem} by
Carlet, Mesnager, Tang, Qi and Pelikaan. LCD codes were originally
introduced by Massey in \cite{massey,massey2}.  These codes have been
studied amongst other things for improving the security of information on sensitive devices
against {\em side-channel attacks} (SCA) and {\em fault non-invasive
  attacks}, see \cite{carlet}, and have found use in {\em data
  storage} and {\em communications' systems}.  %, and consumer

\subsubsection{Notation for convolutional codes}\label{convo} Notation(s) for convolutional codes can be confusing. 
Different equivalent  
 definitions are given in the literature and these are analysed nicely in %in the literature  
\cite{rosenthal1}. %% Interests here  tend towards algebraic definitions
 %% and formulations. 
The following definition is followed here.  A rate $\frac{k}{n}$ convolutional code with parameters $(n,k,\de)$ over a field $\F$ is 
 a submodule of $\F[z]^n$ generated by a reduced
basic matrix $G[z] =(g_{ij}) \in \F[z]^{r\ti n}$ of rank $r$ where $n$ is
the length,  $\de = \sum_{i=1}^k \de_i$ is the {\em degree}
with  $\de_i= \max_{1\leq j\leq k}{\deg g_{ij}}$. Also 
$\mu=\max_{1\leq i\leq r}{\de_i}$ is known as the {\em memory} of the
code and then the code may be given with parameters $(n,k,\de;\mu)$. 
% The notation $(n,r,\de)$ is used for a rate $r/n$ convolutional
% code with  degree $\de$ with the memory not indicated.  
The parameters $(n, r,\delta;\mu, d_f)$ are  used for such a code 
with free (minimum) distance $d_f$. %%  the notation 
%%  $(n,r,\delta;\mu, \geq d_f)$
%% is used to denote such a code with free distance $\geq d_f$.
%% A convolutional code may equivalently be described as
%% follows. A 
%% convolutional code $\C$ 
%% of length $n$ and dimension $k$ is a direct
%% summand of $\F[z]^n$ of rank $k$. See for example \cite{heine} and
%% \cite{rosenthal1}. Here $\F[z]$ is the polynomial
%% ring over $\F$ and  $\F[z]^n = \{(v_1, v_2, \ldots,
%% v_n) : v_i \in \F[z]\}$. 
% Suppose $V$ is a submodule of $\F[z]^n$ and
% that $\{v_1, \ldots, v_r\} \subset \F[z]^n$ forms a generating set for
% $V$. Then  %$V =
% %\text{Im} \, M = 
% %\{uG[z] : u \in \F[z]^r\}$ where 
% $G[z]= \begin{pmatrix} v_1 \\ \vdots \\ v_r \end{pmatrix} \in
% \F[z]_{r\times n}$ is called a {\em generating matrix} of
% $V$.
 Suppose $\C$ is a convolutional code in $\F[z]^n$ of rank $k$. A generating matrix $G[z] \in
\F[z]_{k\times n}$ of $\C$ having rank $k$ is called a
{\em generator} or {\em encoder matrix}  of $\C$. 
A matrix $H \in \F[z]_{n\times(n-k)}$ satisfying $\C = \ker H =
\{v \in \F[z]^n : vH = 0 \}$ is said to be a {\em control matrix} or
    {\em check matrix} of the code $\C$. 

%which may be a known or an unknown quantity. 
%%   and $d_f$ is the
%% free distance of the code. We may also use $(n,r,\delta;\mu)$ to
%% specify the type without stating the free distance. % which may be unknown. 

 %   \quad
    
%% The constructions here use rows or blocks of 
%% unit  schemes.  Series of non-equivalent
%% convolutional codes may be obtained from the same unit system. 

%% Varying the type or model of the unit scheme enables convolutional codes of
%% different types or models to be constructed with specific properties. % are  obtained and analysed  
%% % by utilising special types of  invertible schemes.
%% Using schemes where one of the units has low density allows  the
%% construction of LDPC (low density parity check) convolutional codes,
%% Section \ref{ldpc}. %or other properties. 

%% Specialising to %unitary or
%% orthogonal schemes leads to the construction of series of self-dual and
%% dual-containing convolutional codes, Section \ref{selfdual}. 
%% %% \ref{dualcontain}
%%  Dual-containing  (which includes self-dual) 
%% codes lead to  the construction of quantum codes,
%% \cite{calderbank}; see also \cite{quantum}. 

%% Section \ref{attrib} contains  a  summary with some conclusions which could
%% be read at this stage. % of our methods. 
 Convolutional codes can be {\em
   catastrophic or non-catastrophic}; see
 for example \cite{mceliece} for the basic definitions. A
catastrophic convolutional code is prone to catastrophic error
propagation and is not much use. A
convolutional code described by a generator matrix  with {\em right
polynomial inverse} is a non-catastrophic code; this is sufficient for
our purposes. The designs given here for  the generator matrices allow for
specifying directly the control matrices  and right polynomial
inverses of the generator matrices. % may be constructed directly. 

By Rosenthal and Smarandache, \cite{ros},
 the maximum free distance attainable by an
$(n,r,\delta)$ convolutional  code is  $(n-r)(\floor{\frac{\de}{r}}+1)+
\de +1$. The case  $\delta
=0$, which is the case of zero memory, corresponds to the linear 
Singleton bound $(n-r+1)$. 
The bound $(n-r)(\floor{\frac{\de}{r}}+1)+
\de +1$  
is then called  the {\em generalised Singleton  bound}, \cite{ros}, GSB,
and a convolutional code attaining this bound is known as  an {\em MDS
convolutional code}. The papers \cite{ros} and \cite{smar} 
are major contributions to the area of convolutional codes.   

In convolutional coding theory, the idea of {\em dual code} has 
two meanings. The one used here is what is referred to as the {\em
  convolutional dual code}, see \cite{dualconv} and \cite{dual2}, and
is known also as {\em the module-theoretic dual code}. The other dual
  is called {\em the sequence space dual} code. The two generator
matrices for these `duals' are related by a specific formula.  If
$G[z]H\T[z] = 0$ for a generator matrix $G[x]$, then $H[z^{-1}]z^m$ for memory $m$ generates  the convolutional/(module theoretic) dual code. The code is then dual-containing provided the code generated by $H[z^{-1}]z^m$ is contained in the code generated by $G[z]$.

%is  denoted here by GSB.   

%% In \cite{ros}, Rosenthal and Smarandache give a non-constructive
%% proof of the existence of
%% $(n,r,\de)$ MDS codes %over some finite field $\F$ 
%% and then in
%% \cite{smar} explicit examples over suitably large fields are
%% given. 

%%  The paper \cite{napp} of Napp and Smarandache discusses MDP 
%% (maximum distance profile) convolutional codes and constructs such
%% codes and  \cite{chan} discusses  methods for constructing unit
%% memory MDS convolutional codes. 
 
%% Many of the constructions in the literature are special cases of the general construction here; for example those in \cite{blahut} and \cite{mceliece,mac} are easily obtained by small cases of the general constructione here. 
%% See also \cite{napp} for %some 
%% further developments  on MDS and what are called strongly-MDS  
%% convolutional codes for unit memory codes.

The papers  \cite{heine1},
\cite{rosenthal} introduce  
certain algebraic decoding techniques  for  convolutional
codes. Vetterbi or sequential decoding are available for 
convolutional codes, see \cite{blahut}, \cite{joh} or \cite{mceliece}
and references therein. The form of the control matrix derived here leads to algebraic implementable error-correcting algorithms.

%% In \cite{LCDOrth}, LCD codes are constructed using orthogonal sets of vectors. matrices.  

% methods. %%  these
 
%but these aspects are not dealt with here.   

%% For a given generator matrix 
%% $G[z] \in \F[z]_{r\ti n}$ of a convolutional code,  a codeword is of
%% the form  $u(z)G[z]$ with  $u(z) \in \F[z]^r$. 
%% Here   $u(z)$ is  {\em the information vector } of  the
%% codeword $u[z]G[z]$. 
%% Now  $u(z) = \sum_i\al_iz^i$ for vectors $\al_i\in \F^r$; the $\al_i$ are 
%%   the {\em components} of
%% $u(z)$. The {\em support} of $u(z)$ is the
%% number of non-zero (vector) components of $u(z)$.  

%% Some of the the codes constructed have the property that
%% the maximum bound  is attained by the codeword.
%% for small support of the  information
%% vector  and as the
%% support of the information vector increases the (free) distance of the
%% codeword is increased. %The fields can be relatively small. 
%% (Algebraic decoding techniques 
%% for some convolutional codes will be amplified in \cite{hurley22}.)

 %% The emphasis is on the design and analysis of MDS,
%%  DC MDS codes, MDS QECCs, LCD MDS codes
%%  in both block linear and convolutional form;  these are designed
%%  relative to  Euclidean or Hermitian inner products as required.   
%% %
 % convolutional codes and of MDS  or near
%% MDS dual-containing convolutional codes.

%% design convolutional codes to specific
%% requirements. 
The MDS block linear codes to requirements rate $R$ are
extended to MDS convolutional codes with rate $R$ and with the
order of twice the distance of the linear block MDS codes of the same length. These may
also be specified to be (i) of characteristic a fixed prime $p$, (ii) of
prime order, (iii) DC, or (iv) LCD. Noteworthy here is how
MDS DC block linear  lead to convolutional MDS with memory $1$ LCD codes,  and in
characteristic $2$, LCD MDS block linear codes lead to
the design of DC memory $1$ convolutional MDS codes.  The DC codes 
may be designed with Euclidean inner product or with Hermitian inner product.

\subsubsection{Previously} In  \cite{unitderived} a general method for  deriving MDS
codes to specified rate and specified error-correcting capability is established; 
\cite{hurleyquantum} gives a general method for designing
DC MDS codes of arbitrary rate and error-correcting
capability, from which MDS QECCs can be specified,
\cite{hurleylcd} specifies a general method for designing LCD  MDS codes to arbitrary requirements  
 and \cite{hurleyconv} gives general methods for designing
convolutional codes. The unit-derived
methods devised in \cite{hurleyorder,hurley,hur1,hur2,hur11} are often 
in the background.

\subsection{Abbreviations}
DC: dual-containing\\
LCD: linear complementary dual\\
QECC: quantum error-correcting code\\
MDS: maximum distance separable \footnote{This has different parameter
requirements for
linear block codes, convolutional codes and QECCs.}\\
GSB: Generalised Singleton bound\\
rdist: relative distance, which is the ratio of distance over length.

\section{Summary of design methods}

 Section \ref{thms} describes explicitly the  design algorithms in detail. 
Here's a summary of the main design methods. % acquired.
 %First designs:
%% \begin{enumerate}\label{design} \item  Design MDS  block linear codes
%%   to rate $\geq R$ and distance $\geq (2t+1)$ with efficient decoding
%%   algorithms.  \item  Design MDS  block linear codes
%%   to rate $\geq R$ and distance $\geq (2t+1)$ with efficient decoding
%%   algorithms in fields of (fixed) characteristic $p$. 
%% \item\label{hy}  Design MDS  block linear codes
%%   to rate $\geq R$ and distance $\geq (2t+1)$ with efficient decoding
%%   algorithms in prime order fields. % Algorithm  \ref{alg3}.
%% \end{enumerate}
\subsection{Linear block MDS}Design methods are given initially for:
\\ (i) MDS  linear block codes
  to rate $\geq R$ and distance $\geq (2t+1)$ with efficient decoding
  algorithms. \\ (ii) MDS  linear block codes
  to rate $\geq R$ and distance $\geq (2t+1)$ with efficient decoding
  algorithms over fields of (fixed) characteristic $p$. \\ 
(iii) MDS  linear block codes
  to rate $\geq R$ and distance $\geq (2t+1)$ with efficient decoding
  algorithms over prime order fields. % Algorithm  \ref{alg3}.
%\end{enumerate}

Then  {\em types} of codes are required. Thus
design methods  are obtained, as  in the above (i)-(iii), where   
``MDS linear block'' is replaced by ``MDS linear block {\em of type X}'' where
{\em of type X}  is (a) DC or (b) LCD. These may
be designed with Hermitian inner product when working over fields of type
$GF(q^2)$. For Hermitian inner products, in  item (iii), the `prime order
fields' needs to be replaced by `fields of order $GF(p^2)$, where $p$ is prime'.
The DC codes designed are then used to design MDS QECCs. 

Then   infinite series of such  block block codes
are designed so that the rate approaches 
$R$ and rdist approaches $(1-R)$ for given
  $R$, $0<R<1$;  for DC codes it is required  $\frac{1}{2} < R < 1$. The
infinite series of MDS QECCs
designed from the DC codes have rdist approaching $(2R-1)$ and rate
approaching $R$ for given  $R$, $1>R>\frac{1}{2}$. 
  
  Specifically:
  \begin{itemize}\label{design1} \item  Design of infinite series of
    linear block codes $[n_i,r_i,d_i]$,  such that $\lim_{i\rightarrow
     \infty}\frac{r_i}{n_i}=R$, and  
   $\lim_{i\rightarrow \infty}\frac{d_i}{n_i}=1-R$
   \item Design of infinite series of
    MDS  block linear codes $[n_i,r_i,d_i]$,  such that
   $\lim_{i\rightarrow \infty}\frac{r_i}{n_i}=R$ and
   $\lim_{i\rightarrow
    \infty}\frac{d_i}{n_i}=1-R$ in fields of characteristic $p$. 
  \item Design of infinite series of
    MDS  block linear codes $[n_i,r_i,d_i]$,  such that
    $\lim_{i\rightarrow \infty}\frac{r_i}{n_i}=R$ and  $\lim_{i\rightarrow
    \infty}\frac{d_i}{n_i}=1-R$ in prime order fields or in fields
  $GF(p^2)$ for Hermitian inner product.  
\end{itemize}

Further in the above  a `type'  may be included in the infinite series designed
where   `type' could be  `DC' or `LDC'. For DC, it is necessary that
$R>\frac{1}{2}$, and then infinite series of QECCs
$[[n_i,2r_i-n_i,d_i]]$ are designed where  $\lim_{i\rightarrow
  \infty}\frac{2r_i-n_i}{n_i}=2R-1$ (limit of rates) but still
$\lim_{i\rightarrow 
  \infty}\frac{d_i}{n_i}=1-R$, for given $R$, $R>\frac{1}{2}$

\subsection{Convolutional MDS}
Convolutional MDS codes and infinite series of convolutional MDS codes
are designed. These are designed to specified rate and distance and in
general have  better relative distances. In memory $1$ the distances
obtained are of the order of twice that of the corresponding MDS
linear block 
codes with the same length and rate. Higher memory convolutional codes
are briefly discussed. 
%and will be expanded on later.
%% The methods of \cite{hurleyconv} %% uses
%% %% rows of an invertible 
%% %% matrix to establish a method for designing convolutional code and
%% %% convolutional codes to specific requirements.
%% are relevant  for the convolutional case designs here but the designs here
%% are independent of this. %here.

Thus memory $1$ MDS convolutional
codes are designed as follows.
\begin{itemize} \item Design of MDS  convolutional
  codes of memory $1$ with the same length and rate as the
  corresponding MDS linear block code but with twice
  the distance less $1$. \item Design of MDS convolutional in characteristic $p$ of the same length and
  rate as the corresponding MDS block linear codes but twice the
  distance less $1$. 
  \item  Design of MDS convolutional over a prime field of the same length and
  rate as the corresponding MDS block linear codes but with twice the
  distance, less $1$. 
\end{itemize}
These may be extended to higher memory  MDS
convolutional codes. %and are described later. 

These are then specified for particular types such as DC or LCD.
The LCD linear block codes when
`extended' to convolutional codes give rise to  DC codes, and the DC
linear block codes when
`extended' to convolutional codes give rise to LCD codes in characteristic $2$. 

Infinite series of convolutional codes are designed as follows. 
\begin{itemize} \item Design of infinite series of MDS  convolutional
  codes $(n_i,r_i,n_i-r_i;1,2(n_i-r_i))$  such that 
  $\lim_{i\rightarrow \infty}\frac{r_i}{n_i}=R$, $(2R-1)$,  $\lim_{i\rightarrow
    \infty}\frac{d_i}{n_i}=(2R-1)$.
\item Design of such infinite series over fields of (fixed) characteristic $p$.
\item Design of such infinite series over fields of prime order.  
\end{itemize}

%% These are  MDS memory $1$ convolutional codes.
Then such  infinite series designs of convolutional codes are
described  for `types' of codes such as DC or LCD or  %%  with appropriate
%% modifications as necessary.
QECCs, which may be  designed from DC codes.
 \subsection{Characteristic $2$ and prime fields}

The designs over the fields $GF(2^i)$ and over prime fields $GF(p)=\Z_p$  have a
particular interest and have nice properties. The 
designs to specific requirements and specific type both linear block
and convolutional can be constructed over such fields. 
%%  $GF(2^i)$ has an element of order
%% $2^i-1$ whose powers describe all its non-zero elements.
$GF(p)$, $p$ a prime, has an element of order
$(p-1)$ which is easily found and arithmetic within $GF(p)$ is {\em modular
arithmetic}. %% For Hermitian inner products, it is necessary to work over
%% the fields  $GF(2^{2s})$ and over $GF(p^2)$.

\subsection{Examples} 
Examples are given throughout. Although there is no restriction on
length in general, examples explicitly written out  here are
limited in their size. 
%% is some emphasis in going from designing block linear MDS, MDS
%% dual-containing, 
Example \ref{pro7} below is a prototype example of small order which
has some hallmarks of the general designs; it is instructional  and may be read now with
little preparation. Example \ref{large} is an instructional example on the
design methods for  linear block (MDS)  codes to meet specific rates and distances. The
general designs are more powerful and include designs for convolutional
codes. %%  there are
%% others indeed. 
%% to see how the designs are instigated;  this is  a low length case for
%% demonstation purposes  although in general there is no restriction on
%% length and rate in the general setup. %%  this may be 
%% considered a prototype example,  

          %% = \begin{ssmatrix}e_2\\e_5\\e_3\\e_4 \end{ssmatrix}
          %% + \begin{ssmatrix} 0 \\ e_0 \\ e_1 \\ e_2\end{ssmatrix}z
          %%   = A+Bz$, say.
          %%   Then $(A+Bz)*\{(f_0,f_1,f_6) - (f_5,f_3,f_4)z\} =
          %%   0$. The Euclidean dual matrix is $\begin{ssmatrix}f_0^T \\ f_1^T
          %%     \\ f_6^T\end{ssmatrix}
          %%     - \begin{ssmatrix}f_5^T\\f_3^T\\f_4^T\end{ssmatrix}z^{-1}=\begin{ssmatrix}e_0\\e_6\\e_1\end{ssmatrix}
          %%     - \begin{ssmatrix}e_2 \\ e_4\\e_3\end{ssmatrix}z^{-1}$

          %%       Since we are in characteristic $2$, note that
          %%       $-=+$. The dual matrix is equivalent to
          %%       $\begin{ssmatrix}e_2\\e_4\\e_3\end{ssmatrix}z 
          %%         + \begin{ssmatrix}e_0 \\ e_6\\e_1\end{ssmatrix}$.
          %%           It is now necessary to show that the code
          %%           generated by $\begin{ssmatrix}e_2\\e_4\\e_3\end{ssmatrix}z 
          %%         + \begin{ssmatrix}e_0 \\ e_6\\e_5\end{ssmatrix}$ is
          %%           contained in the code generated by $A+Bz$. 
%\end{example}
%% \subsection{Abbreviations}
%% DC: dual-containing\\
%% LCD: linear complementary dual\\
%% QECC: quantum error-correcting code\\
%% MDS: maximum distance separable \footnote{This has different parameter
%% requirements for
%% linear block codes, convolutional codes and QECCs.}\\
%% GSB: Generalised Singleton bound\\
%% rdist: relative distance, which is the ratio of distance over length.

\section{Constructions} 
%Here is a summary of previous designs. 
 The following constructions follow from \cite{unitderived}; 
 see also  \cite{hurleyorder}, \cite{hurley}.

\begin{construction}\label{four} Design MDS linear block codes. 
  \quad
  {\em 
\begin{itemize}{\label{four1}} \item Let $F_n$ be a Fourier $n\ti n$
  matrix over a finite field.
\item Taking $r$ rows of $F_n$ generates an $[n,r]$ code. A
  check matrix for the code is obtained by eliminating the
  corresponding columns of the inverse of $F_n$. 
  \item\label{wrap} Let $r$ rows of $F_n$ be chosen in arithmetic sequence such
    that the arithmetic difference $k$ satisfies $\gcd(k,n)=1$. The
    code generated by these rows is an MDS $[n,r,n-r+1]$ code. There
    exists an explicit efficient decoding algorithm of $O(\max\{n\log
    n, t^2\})$,  $t=\floor{\frac{n-r}{2}}$,    $t$ is the
    error-correcting capability of the code. 

      In particular this is true when $k=1$,  that is when the rows are taken
      in sequence.  \end{itemize}
}
\end{construction}

General Vandermonde matrices may be used instead of Fourier matrices
but are not necessary.  
For a given Fourier $n\ti n$ matrix $F_n$ under consideration the rows of $F_n$
in order are 
denoted by $\{e_0,e_1, \ldots, e_{n-1}\}$ and $n$ times the columns in
order are denoted by $\{f_0,f_1,\ldots, f_{n-1}\}$. $F_n$ is generated by
a primitive $n^{th}$ root of unity $\om$; thus $\om^n=1$, $\om^r\neq
1, 0<r<n$. Hence $e_i=(1,\om^i,\om^{2i},\ldots,
\om^{i(n-1)})$. 
Indices are taken modulo $ n$ so that $e_{i+n}=e_i$. The arithmetic sequences
in Construction \ref{wrap} may wrap around; for example when $n=10$,  such
arithmetic sequences include $\{e_8,e_9,e_0,e_1\},
\{e_3,e_6,e_9,e_2,e_5\}$.

Note also that if $B$ is a check matrix then so also is $nB$ for any
$n\neq 0$. Thus in the above Construction \ref{four} the check matrix
may be obtained from $n$ times the columns of the inverse.

\begin{construction}\label{dualconstruct} Design DC MDS linear block codes.

  {\em This construction follows from  \cite{hurleyquantum}.
\begin{itemize}
\item Let $F_n$ be a Fourier $n\ti n$ matrix with rows $\{e_0,\ldots,
  e_{n-1}\}$ in order and $n$ times the columns of the inverse in order
  are denoted by $\{f_0,\ldots,f_{n-1}\}$.
\item Let $r> \floor{\frac{n}{2}}$ and define
  $A= \begin{ssmatrix}e_0\\e_1\\ \vdots \\e_{r-1} \end{ssmatrix} $.
  \item A check code for the code generated by  $A$
    is $B=(f_r,f_{r+1},\ldots, f_{n-1})$.
    \item Then
    $B^T= \begin{ssmatrix}f_r^T \\f_{r+1}^T \\ \vdots
      \\ f_{n-1}^T\end{ssmatrix} = \begin{ssmatrix}
        e_{n-r}\\ e_{n-r-1} \\ \vdots \\ e_1 \end{ssmatrix}$. Hence
the code generated by $A$ is a DC $[n,r,n-r+1]$ MDS code.
\item This works for any $r$ such that $n> r> \floor{\frac{n}{2}}$.
\end{itemize}}
\end{construction}

\begin{construction}\label{lcdconstruct} Design LCD MDS linear block  codes.

{\em   The design technique  is taken  from \cite{hurleylcd}.

  \begin{itemize}
  
\item Construct a Fourier $n\ti n$ matrix. Denote its  
  rows in order by $\{e_0,\ldots, e_{n-1}\}$ and $n$ times the
  columns of the inverse in order are denoted by $\{f_0,\ldots,f_{n-1}\}$.
\item Design  $A$ as follows. $A$ consists of first  row $e_0$ and 
  rows $\{e_1,e_{n-1},e_2,e_{n-2}, \ldots, e_r,e_{n-r}\}$ for $r\leq
  \floor{\frac{n}{2}}$. ($A$ consists of $e_0$ and pairs
  $\{e_i,e_{n-i}\}$ starting with $\{e_1,e_{n-1}\}$.)
  \item Set $B\T=(f_{r+1},f_{n-r-1}, \ldots, f_{\frac{n-1}{2}},f_{\frac{n+1}{2}})$
  when $n$ is odd and \\ $B\T=(f_{r+1},f_{n-r-1}, \ldots,
  f_{\frac{n}{2}-1},f_{\frac{n}{2}+1}, f_{\frac{n}{2}})$ when $n$ is even.
  \item Then $AB\T=0$ and
  $B$ generates the dual code $\C^\perp$ of the code $\C$ generated by
    $A$.
    \item Using 
      $f_i\T=e_{n-i}$ it is easy to check that $\C\cap \C^\perp = 0$.
      \item The rows of $A$ are in sequence $\{n-r, n-r+1, \ldots, n-1, 0, 1,
  \ldots, r-1\}$ and so $A$ generates an MDS LCD linear block  $[n,2r+1, n-2r]$
  code.  
  \end{itemize} }
\end{construction}
%% When $\gcd(p,n)=1$ then $p^{\phi(n)}\equiv 1 \mod n$. The least
%% positive power such that $p^s \equiv 1 \mod n$ is the order of $p
%% \mod n$ and write $\ord(p,n) =s$.

{\em The general method of constructing MDS codes by choosing rows
      from Fourier matrices does
      not take into account the power of the other non-chosen
      rows.} This can be remedied by going over to convolutional
codes. The convolutional codes obtained carefully in this way
are MDS convolutional codes with free
      distance of the order of twice the distance of the
      corresponding MDS linear block code with the same length and
      rate. 

\begin{lemma}\label{weight} Let $F$ be a Fourier $n\ti n$ matrix with
  rows $\{e_0,e_1,\ldots,e_{n-1}\}$. Define  $A=\begin{ssmatrix}e_0\\e_1\\ \vdots
\\ \vdots 
\\ e_{r-1}\end{ssmatrix}, B=\begin{ssmatrix}0\\ \vdots
\\0\\e_r\\ \vdots \\e_{n-1}\end{ssmatrix}$ where $r>
\floor{\frac{n}{2}}$ and the first $(2r-n)$ rows of $B$ consist of
zeros. Let $P$ be a non-zero vector of length $n$. Then $\wt P(A+Bz)
\geq 2(n-r)+1$. 
\end{lemma}
\begin{proof} $PA$ has $\wt \geq (n-r+1)$ as $A$ generates an
$[n,r,n-r+1]$ code. $PB$ has weight $\geq r+1$ except when $P$ 
has the last  $(n-r)$ entries consisting of zeros, as the non-zero rows of
$B$ generate an $[n,n-r,r+1]$ code. Now $(n-r+1) + r+1 = n+2 >
2(n-r)+1$. When $P$ has last  $(n-r)$ entries
consisting of zeros then $PA$ contains a non-zero sum of $\{e_0,e_1,\ldots,
e_{2r-n}\}$ which is part of an $[n,2r-n, 2r-n-n+1] = [n,2r-n,2(n-r)+1]$
code and so has weight $\geq 2(n-r)+1$ as required.
\end{proof}

In fact if $P$ is a polynomial of degree $t$ then $\wt P(A+Bz)\geq
2(n-r)+ t+1$ so weight increases with the degree of the  multiplying
polynomial vector.   

Lemma \ref{weight} may be generalised as follows.
\begin{lemma} Let $F$ be a Fourier $n\ti n$ matrix with
  rows $\{e_0,e_1,\ldots,e_{n-1}\}$. Let $A$ be chosen by taking $r$ rows, $r>
\floor{\frac{n}{2}}$, of the Fourier
$n\ti n$ matrix in arithmetic sequence with arithmetic difference $k$
satisfying $\gcd(k,n)=1$. Let $B$ be the matrix with first $(2r-n)$
rows consisting of zeros and the last $(n-r)$ rows consisting of the
rest of the rows of $F$ not in $A$; these last rows of $B$ are also in sequence with
arithmetic difference $k$ satisfying $\gcd(k,n)=1$. Let $P$ be any
non-zero vector of length $n$. Then $\wt P(A+Bz) \geq 2(n-r)+1$.
\end{lemma}

Before describing the general design methods, it is instructional to
consider the following example.
%% The following example demonstates some of the general principles and
%% design techniques involved. This  example though has small length
%% and  in general there is no restriction on the length and rate that can be achieved
%% in designing the MDS codes.
See also Example \ref{large} below for a
larger example demonstrating the design techniques for constructing a
code  to given rate and distance.

When $\gcd(p,n)=1$, $\ord(p,n)$ denotes the least positive power $s$
such that $p^s\equiv 1\mod n$. 
%and almost always the maximum distance possible is attained. 

\begin{example}\label{pro7} Consider $n=7$. Now $\ord(2,7)=3$ so Fourier
  $7\times 7$ matrix may be constructed  over $GF(2^3)$. The Fourier
  $7\ti 7$ matrix may also be constructed over $GF(3^6)$ as
  $\ord(3,7)=6$ and over many other fields whose characteristic does
  not divide $7$. It may be formed over the prime field $GF(29)$ as
  $\ord(29,7) = 1$; arithmetic in $GF(29)$ is then modular arithmetic.

 We'll stick to $GF(2^3)$ for the moment; when  Hermitian inner
 product is required we'll move  to $GF(2^6)$.

 The rows in order of a Fourier $7\ti
 7$ matrix under consideration are denoted by
 $\{e_0,e_1,e_2,\ldots, e_6\}$ and ($7$ times) the columns of  the inverse in
 order are denoted by $\{f_0,\ldots, f_6\}$; note $7=1$ in
 characteristic $2$. Then $e_if_j=\delta_{ij}, e_i^T = f_{7-i},f_i^T=e_{7-i}$.
 \begin{enumerate}
   \item 
 Construct  $A
  = \begin{ssmatrix}e_0 \\ e_1 \\ e_2\\e_3\end{ssmatrix}$. Then
    $A*(f_4,f_5,f_6)= 0 $. The Euclidean dual matrix is
      $(f_4,f_5,f_6)^T
      = \begin{ssmatrix}f_4^T\\f_5^T\\f_6^T\end{ssmatrix}
        = \begin{ssmatrix}e_3\\e_2\\e_1\end{ssmatrix}$. Thus the code
          generated by $A$ is a DC code $[7,4,4]$ code.
\item 
          To obtain a DC code relative to the Hermitian
          inner product work in $GF(2^6)$. Again the rows of the 
          Fourier $7\ti 7$ matrix over $GF(2^6)$ are denoted by $\{e_0, \ldots,
          e_6\}$. Here $e_i^l=e_{il}$ as explained below where
          $l=2^3$ and thus since $2^3 \equiv 1 \mod 7$ it follows
          that $e_i^l=e_{il} = e_i$. Thus the code generated by $A
  = \begin{ssmatrix}e_0 \\ e_1 \\ e_2\\e_3\end{ssmatrix}$ is an
    Hermitian DC MDS $[7,4,4]$ code. 

      \item      A  DC MDS convolutional code over $GF(2^3)$
           and a DC Hermitian MDS convolutional code over $GF(2^6)$
           are obtained as follows. The distance obtained is
           of the order of twice the distance of the corresponding MDS
           linear block code. %% This shows the power and
           %% extent of the constructions.
           %% a dual-containing convolutional
          %% code which attains the GSB and is MDS.
%% Here's an example in $GF(2^3)$ and in $GF(2^6)$ for
%% Hermitian case to show the power and extent of the constructions.

%% Let $F=F_7$ denote the Fourier matrix over $GF(2^3)$ or over
%% $GF(2^^6)$; the $F$ is different for each case but still denote the
%% rows by $e_0,\ldots,e_7$ in order and $7$ times the columns of
%% the inverse of $F$ by $f_0,f_1,\ldots,f_7$. 

%% Desing $A=\begin{ssmatrix}e_0\\e_1\\e_2\\e_3\end{ssmatrix}$ 
%% $B=(f_4,f_5,f_6)$. Then $AB=0$. Now $B^T=\begin{ssmatrix}e_3
%% \\e_2\\e_1\end{ssmatrix}$. Thus the code generated by $A$ is a
%% dual-containing $[7,4,4]$ code over $GF(2^3)$ and an Hermitian dual-containing $[7,4,4]$  code over $GF(2^6)$.

(In characteristic $2$, $-1=+1$.)

\item Now design  $G[z]= \begin{ssmatrix}e_0\\e_1\\e_2\\e_3\end{ssmatrix}
+ \begin{ssmatrix}0\\e_4\\e_5\\e_6\end{ssmatrix}z$. Then 
$G[z]*((f_4,f_5,f_6) +(f_1,f_2,f_3))z = 0$,  $H\T[z]= (f_1,f_2,f_3))z$.
Then $H[z^{-1}] = \begin{ssmatrix}e_3\\e_2\\e_1\end{ssmatrix}
+ \begin{ssmatrix} e_6\\e_5\\e_4\end{ssmatrix}z^{-1}$.
Thus a control matrix is $K[z]= \begin{ssmatrix}e_3\\e_2\\e_1\end{ssmatrix}z
+ \begin{ssmatrix} e_6\\e_5\\e_4\end{ssmatrix}$. It is easy to show
that the convolutional code generated by $K[z]$ has trivial
intersection with the convolutional code generated by $G[z]$. Thus the
convolutional code generated by $G[z]$ is a LCD  $(7,4,3;1,d_f)$
code. The GSB for a code of this form is 
$3(\floor{\frac{3}{4}+1}+ 3+1 = 7$. The free distance of the one
constructed may be shown to be $7$ directly or from the general Lemma
\ref{weight} below. 

\item Starting with the MDS DC $[7,4,4]$ code, 
a corresponding convolutional code $(7,4,3;1,7)$ is designed
of memory $1$ which is LCD and has almost twice the distance of the
DC linear block code. 

\item Is it possible to go the other way? Methods for designing MDS LCD
linear block codes are established in \cite{hurleylcd}. Following the
method of \cite{hurleylcd}, let $A
= \begin{ssmatrix}e_0\\e_1\\e_6\\e_2\\e_5\end{ssmatrix}$ and hence 
  $A*(f_3,f_4) = 0$. Then $(f_3,f_4)\T
  = \begin{ssmatrix}e_4\\e_3\end{ssmatrix}$ giving that the code
    generated by $A$ is an MDS $[7,5,3]$ LCD code. Then
  $G[z]= \begin{ssmatrix}e_0\\e_1\\e_6\\e_2\\e_5\end{ssmatrix}
    + \begin{ssmatrix}0\\0\\0\\e_3\\e_4\end{ssmatrix}z = A+Bz$, say,
      gives a convolutional $(7,5,2;1,5)$ MDS code. A control matrix
      is $H\T[z] =
      (f_3,f_4)-(f_2,f_5)z$. $H[z^{-1}]= \begin{ssmatrix}e_4\\e_3\end{ssmatrix}
        + \begin{ssmatrix}e_5\\e_2\end{ssmatrix}z^{-1}$. Thus the dual
          code has generating matrix
          $\begin{ssmatrix}e_4\\e_3\end{ssmatrix}z
            + \begin{ssmatrix}e_5\\e_2\end{ssmatrix}$. Now it is
              necessary to show that the code generated by $G[z]$ is
              DC.

Note $\begin{ssmatrix}0&0&0&0&1
  \\0&0&0&1&0\end{ssmatrix}*\{\begin{ssmatrix}e_0\\e_1\\e_6\\e_2\\e_5\end{ssmatrix}+ \begin{ssmatrix}0\\0\\0\\e_3\\e_4\end{ssmatrix}z\}
    =\begin{ssmatrix}e_5\\e_2\end{ssmatrix}+\begin{ssmatrix}e_4\\e_3\end{ssmatrix}z$.
    Hence the code generated by $G[z]$  is DC over
    $GF(2^3)$ and is Hermitian DC over $GF(2^6)$.

  \item Construct $G[z] = \begin{ssmatrix} e_0\\e_1\\e_2 \end{ssmatrix}
    + \begin{ssmatrix}0\\e_3\\e_4\end{ssmatrix}z + \begin{ssmatrix}
        e_5\\0\\e_6 \end{ssmatrix}z^2 $

        Then $G[z]$ is a convolutional code of type $(7,3,5;2)$; the
        degree is $5$. The GSB for such a code is
        $(7-3)(\floor{\frac{5}{3}+1} + 5+1 = 4*2 + 5+1 = 14$. In fact
        the free distance of this codes is actually $14$. This may be
        shown in an analogous way to the proof of Lemma \ref{weight}.

        $G[z]*((f_3,f_4,f_5,f_6) - (f_1,f_2,0,0) - (0,0,f_0,f_2)z^2))=
        0$, $G[z]*((f_0,f_1,f_2)) = 7I_3$.

        The result is that the code generated by $G[z]$
        is a non-catastrophic convolutional
        MDS  $(7,3,5;2,14)$ code. Note the free distance
        attained is $5*3-1$ where $5$ is the free distance of a
        $[7,3,5]$ MDS code; the distance is tripled less $1$. This is
        a  general principle for the more general cases -- the free
        distance is of order three  times the distance of the same
        length and dimension MDS code. 

     It's not a dual-containing code nor a LCD code. To get such
        codes requires a compromise on the distance.  $G[z]
        = \begin{ssmatrix} e_0\\e_1\\e_2 \end{ssmatrix} 
    + \begin{ssmatrix}0\\e_3\\e_4\end{ssmatrix}z + \begin{ssmatrix}
        0\\e_5\\e_6 \end{ssmatrix}z^2 $. This give $(7,3,4;2)$
      convolutional code which turns out to be an LCD code but the
      free distance is only $7$. The GSB for such a code is $13$. 

\item Now for memory $3$ define $G[z]
      = \begin{ssmatrix}e_0\\e_1\end{ssmatrix}+ \begin{ssmatrix}e_2\\e_3\end{ssmatrix}z+\begin{ssmatrix}e_4\\e_5\end{ssmatrix}z^2
          + \begin{ssmatrix}0\\e_6\end{ssmatrix}z^3$
            The GSB for such a code is
            $(7-2)(\floor{\frac{5}{2}}+1)+5+1 = 21$. The free distance
            of the code is actually $21$ so the code is a
            $(7,2,5;3,21)$ convolutional MDS code. This is $6*4-3$
            where $6$ is the free distance of the corresponding block
            linear MDS code $[7,2,6]$.

      \item      Ultimately get a convolutional code
            $e_0+e_1z+e_2z^2+e_3z^3+e_4z^4+e_5z^5+e_6z^6$ which is the
            convolutional MDS code $(7,1,6;6,47)$ code which is
            repetition convolutional code.  
    \end{enumerate}
\end{example}

\section{Specify the codes}

\subsection{Matrices to  work and control}
Many  of the designs hold using general  Vandermonde matrices but the Fourier
matrix case is considered for clarity.

If the Fourier $ n\times n $
matrix $F_n$ exists over a finite field 
then the characteristic $p$ of the field does not divide $n$  which
happens if and only if $\gcd(p,n)=1$.
%Denote by $F_n^*$ the matrix such that $F_nF_n^*=nI_n$. 

Let $F_n$ denote a Fourier matrix of size $n$. Over which finite
fields precisely may this matrix be constructed? Suppose $\gcd(p,n)=1$. 
Then $p^{\phi(n)}\equiv 1 \mod n$, where $\phi$ is the Euler $\phi$-function. Hence  there exists a least positive
power  $s$ that $p^s \equiv 1 \mod n$; this $s$ is called {\em the order of
$p$ modulo $ n$}. Use $\ord (p,n)$ to denote the order of $p \mod n$.
\begin{lemma}
  Let $p$ be any prime such that $p\not|n$ and $s=\ord (p,n) $.
  
(i) There exists an element of order  $n$ in $GF(p^s)$ from which the
Fourier $n\ti n$ matrix may be constructed over $GF(p^s)$.

(ii) The Fourier $n\ti n$ matrix cannot exist
      over a finite field of characteristic $p$ of order smaller than
      $GF(p^s)$.

      (iii) There exists a  Fourier $n\ti n$ matrix 
      over any $GF(p^{rs}), r\geq 1$ and in particular over $GF(p^{2s})$.

\end{lemma}\label{order}

\begin{proof} (i) There exists an element of order $p^s-1$ in
  $GF(p^s)$, that is for some $\omega \in GF(p^s)$, $\omega^{p^s-1}
  =1 $. Now $(p^s -1) = qn$ for some $q$ and so $(\omega^q)^n = 1 $ in $GF(p^s)$, giving an
element of order $n$ in $GF(p^s)$. This element may then be used to
construct the Fourier $n\times n$ matrix over $GF(p^s)$.

Proofs of (ii) and (iii) are omitted.
\end{proof} 
%% Over which fields can the Fourier matrix exist? Suppose then
      %% $p\not| n$. Then $\ord(p,n)= s$ for some $s$. Thus $GF(p^s)$
      %% contains an element of order $n$ and the Fourier $n\ti n$ matrix
      %% exists over $GF(p^s)$

For a vector $v=(v_1,v_2,\ldots,v_r)$ define $v^l=(v_1^l,v_2^l,\ldots,
v_r^l)$.

The following lists some properties of a Fourier matrix of size
$n$ over a finite field. These are used throughout. 
\begin{enumerate}
\item Let $F_n$ be a Fourier $n\ti n$ matrix over a field generated by $\om$,
  where $\om^n=1$ and $\om^r\neq 1$ for $0<r<n$.
 \item Denote the rows of  $F_n$  by $\{e_0,e_1,\ldots,
e_{n-1}\}$ in order and
$n$ times the columns of the inverse of $F_n$ in order by
$\{f_0,f_1,\ldots,f_{n-1}\}$.
\item Then $e_if_j=\delta_{ij},
e_i\T=f_{n-i},f_i\T = e_{n-i}$. %These are fundamental relations.
\item 
%For a vector $v=(v_1,v_2,\ldots,v_k)$ define $v^l=(v_1^l,v_2^l,\ldots,
%v_k^l)$. %% The rows of a Fourier $n\ti n$ matrix are $e_i=(1,\om^i,\om^{2i},
%% \ldots, \om^{(n-1)i})$ and so $e_i^l=(1,\om^{il},\om^{2il},\ldots,
%% \om^{(n-1)il})= e_{il}$ where all indices are taken modulo $n$.  
The rows of the $F_n$  are given by
$e_i=(1,\om^i,\om^{2i},\ldots, \om^{(n-1)i})$. Indices are  to be taken modulo $n$. 
 \item $e_i^l=(1,\om^{il},\om^{2il},\ldots, \om^{(n-1)il}) =
   e_{il}$.
   \item Note that if  $l\equiv 1 \mod n$ then $e_i^l = e_{il}=e_i$.
\end{enumerate}

%% which are used a
%% number of times.

       Within $GF(p^{2s})$ the Hermitian inner product is
      defined by ${\langle u,v \rangle}_H= {\langle u,v^l\rangle}_E$ where $l=p^s$. In this setup $\langle e_i,
      e_j \rangle_H =  \langle e_i, e_j^l\rangle_E = \langle
      e_i,e_{jl}\rangle_E$. %%  = \langle e_i, e_j\rangle$ as $l = p^s \equiv 1 \mod 
%% n$. 
This facilitates  the construction of Hermitian
inner product codes  over $GF(2^{2s})$.

\begin{example}\label{tut}

  %This is a prototype example.

  Consider length $n=10$. Now  $\ord(3,10) = 4$.  Construct the Fourier
  $10\times 10 $ matrix $F_{10}$ over $GF(3^4)$. Denote the rows in
  order of $F_{10}$ by $\{e_0,e_1,\ldots, e_9\}$ and the $10$ times
  the columns of the inverse of $F_{10}$  in order by $\{f_0,f_1,
  \ldots, f_9\}$. 
  Then $e_if_j=\delta_{ij}$. Also $e_i\T = f_{10-i}$, %% but in
  %% characteristic $3$, $10=1$ and so $e_i^T = f_{10-i},
  $f_i\T=e_{10-i}$.

\begin{enumerate} \item  As in \cite{hurley} construct
  $A= \begin{ssmatrix}e_0\\e_1\\e_2\\e_3\\e_4\\e_5\end{ssmatrix}$. Then
    $A*(f_6,f_7,f_8,f_9) = 0$. Now by Construction \ref{four},
    \cite{unitderived} 
    the code generated by $A$ is an MDS $[10,6,5]$ code. \item The Euclidean
    dual of $A$ is  $(f_6,f_7,f_8,f_9)\T
    = \begin{ssmatrix}f_6\T\\f_7\T\\f_8\T\\f_9\T\end{ssmatrix}
      = \begin{ssmatrix}e_4\\e_3\\e_2\\e_1\end{ssmatrix}$. Thus $A$ is
        a DC code.

\item         Now $e_i^{3^2} =e_{9*i}=e_{-i}=e_{10-i}$. Thus $A$ is not a
        DC code under the Hermitian inner product induced 
        in $GF(3^4)$.

        Consider $GF(3^8)$. This has an element of order $10$, as
        $3^8-1 = (3^4-1)(3^4+1)$,  
        and so the Fourier $10\times 10$ matrix may be constructed
        over $GF(3^8)$. Here then the Hermitian inner product is
        ${\langle u, v \rangle}_H = 
        {\langle u,v^{3^4}\rangle}_E$, where the suffix $E$ denotes
        the Euclidean inner product.

      Now $e_i^l=e_{il} = e_{i*(3^4)}=e_i$. Thus ${\langle e_i,
          e_j \rangle}_H = {\langle
          e_i,e_j\rangle}_E$. Hence here then the code generated by
        $A$, constructed in
        $GF(3^4)$,  is a
        DC  code under the {\em Hermitian inner product}.

        \item Also $\ord(7,10)=4$ so the above works over $GF(7^4)$ and
        over $GF(7^8)$ when seeking Hermitian DC  codes.

\item   Better though is the following. $\ord(11,10)=1$ and so the
        prime field $GF(11)$ may be considered. $A$ as above is then a
        DC code over the prime field $GF(11)$ and a DC
        Hermitian code when considered over $GF(11^2)$.

       \item  What is an element of order $10$ in $GF(11)$?  In
        fact $\omega = (2 \mod 11)$ works. In $GF(11)$ the arithmetic
        is modular arithmetic. An element of order $10$ is required in
        $GF(11^2)$. Now $GF(11^2)$ is constructed  by
        finding an irreducible polynomial of degree $2$ over
        $Z_{11}$. A primitive element is easily found.

        %% In $F_{10}$ the linear $[10,6,5]$ MDS code is obtained by
        %% considering $A = \begin{ssmatrix}
        %%   e_0\\e_1\\e_2\\e_3\\e_4\\e_5\end{ssmatrix}$. Then letting
        %%   $B=(f_6,f_7,f_8,f_9)$ gives that $AB=0$ and the check matrix
        %%   is obtained. The full rows of $F_{10}$ may be used to obtain
        %%   a convolutional MDS code of the same rate but with twice the
        %%   distance less $1$.

      \item Let $A$ be as above and $B=\begin{ssmatrix} 0
          \\0\\e_6\\e_7\\e_8\\e_9\end{ssmatrix}$. Define $G[z]=A+Bz$.
          The code $\C$ generated by $G[z]$ is a $(10,6,4;1,d_f)$
          convolutional code. 
         \item  Then $G[z]*\{(f_6,f_7,f_8,f_9) - (f_2,f_3,f_4,f_5)z\} = 0$ and
          $G[z]*\{f_0,f_1,f_2,f_3,f_4,f_5)\} = I_5$. Thus $G[z]$ is a
           non-catastrophic generator for the code.
           \item The GSB of such a code is
          $(10-6)(\floor{\frac{4}{6}})+4 + 1 = 4+4+1 = 9$. The free
          distance of $\C$  may be shown, using Lemma \ref{weight}
          essentially,  to be $9$ and so is thus 
          an MDS convolutional $(10,6,4;1,9)$ code.
          %% \item Note that the
          %% linear code generated by $A$ has distance $4$ and the
          %% non-zero vectors in $B$ generate a linear code of distance
          %% $7$. The free distance may be proved by looking at
          %% $(P_0+P_1z+\ldots )(A+Bz)$. The distance $9$ is obtained
          %% by $P(A+Bz)$ when $P=(\al_0,\al_1,0,0,0,0)$ and then $P*A$
          %% involves just $\{e_0,e_1\}$ and thus has distance $\geq 9$ as
          %% its a codeword in a  $[10,2,9]$ code.
          \end{enumerate}
    Since $\ord(11,10)=1$ the calculations may be done over the prime
    field $GF(11)$,  and over $GF(11^2)$ when Hermitian codes are
    required.       
\end{example}

%\section{Vandermonde/Fourier matrices}

\begin{example}Consider $n=2^5-1 = 31$. Then $\ord(2,31)= 5$ and so
  the Fourier $31\ti 31$ matrix may be constructed over $GF(2^5)$ but
  also over $GF(2^{10})$. Let $r=17$, and $A$ consist of rows $\langle{e_0, e_1, \ldots,
    e_{16}}\rangle$ and $ B= (f_{17},f_{18}, \ldots, f_{30})$. Then $A$
    generates an $[31,17,15]$ MDS code. Now $B\T$ consists of rows
    $\{e_{14}, e_{13}, \ldots, e_{1}\}$, using $f_i\T = e_{31-i}$. Thus The code
    generated by $A$ is a DC MDS $[31,17,15]$ code Euclidean
    over $GF(2^5)$ and Hermitian over $GF(2^{10})$.
\end{example}

%% \begin{example}

%%   Example of  constructing  MDS and MDS DC codes over a prime field.

%% Consider $GF(11)$. This has an element of order $10$. For example
%% $\ord(2,11) = 10$ and thus $(2 \mod 11)$ may be used to construct the
%% Fourier $10\ti 10$ matrix $F$ over $GF(11)=\Z_{11}$.
%% Denote the rows of $F$ by $\{e_0,e_1,\ldots,e_{9}\}$.
%% \begin{enumerate}
%% \item $A= \begin{ssmatrix} e_0 \\e_1\\e_2\\e_3\\e_4\\e_5\end{ssmatrix}$
%% generates a DC MDS code $[10,6,5]$. A check matrix is
%% $H=(f_6,f_7,f_8,f_9)$ and $H\T = \begin{ssmatrix} e_4
%%   \\e_3\\e_2\\e_1\end{ssmatrix}$.
%%   \item
%%     $B= \begin{ssmatrix}e_0\\ e_1\\e_9\\e_2\\e_8\\e_3\\e_7\end{ssmatrix}
%%     $ generates an LCD $(10,7,4)$ code.
%%     \item Use an element of order $10$ to generate the Fourier $10\ti
%%       10$ matrix over $GF(11^2)$. Then $A$ as described in item 1. but
%%       now over $GF(11^2)$ generates a DC MDS $[10,7,4]$ code
%%       Hermitian code. 
%%     \item The code $A$ in the last item can be used to form an
%%       Hermitian MDS $[[10, 4, 4]]$ QECC. 
%% \item The code $A+Bz$ with
%%   $B=\begin{ssmatrix}0\\0\\e_6\\e_7\\e_8\\e_9 \end{ssmatrix}$
%%   is an MDS convolutional  $(10,6,4;1,9)$ code over $GF(11)$ or when
%%   the $e_i$ are formed in $GF(11^2)$ is an MDS convolutional Hermitian
%%   $(10,6,4;1,9)$ code in $GF(11^2)$.
%% \end{enumerate} 
%% The arithmetic is also modular arithmetic in $GF(11)$. The arithmetic
%% in $GF(11^2)$ is more complicated but not much; an element of order
%% $10$ is required within $GF(11^2)$.
%% \end{example}

      \begin{example}\label{large} {Example of a general technique.}

      It is required to construct a rate
        $\geq \frac{7}{8}$ codes which can correct $25$ errors; thus a
         distance $\geq 51$ is required.
{\em LINEAR BLOCK:
         \begin{enumerate}
\item         An $[n,r,n-r+1]$ type code with $\frac{r}{n}\geq
         \frac{7}{8}=R$ and $(n-r+1) \geq 51$ is required. Thus
         $n-r\geq 50$ giving 
         $n(1-R)\geq 50$ and so it is required that $n\geq 400$.
     \item     Construct Fourier matrix $F_{400}$ of size $400\ti 400$ over some
         suitable field, to be determined. The rows are denoted by
         $\{e_0,\ldots, e_{399}\}$ and the columns of $400$ times the
         inverse by $\{f_0,\ldots, f_{399}\}$.

\item          Define $A$ to be the matrix with rows
         $\{e_0,\ldots,e_{349}\}$. Then by \cite{hurleyquantum} $A$ is a
         DC $[400,350,51]$ MDS code. By the CSS
         construction a  $[[400,300,51]]$ MDS QECC  is designed.
        \end{enumerate}
        \begin{itemize}
      \item    Over which fields can $F_{400}$ be defined? The characteristic
         must not divide $400$ but otherwise the fields can be
         determined by finding $\ord(p,n)$ where $\gcd(p,n)=1$.
         \item Now
         $\ord(3,400)=20,\, \ord(7,400)=4, \, \ord(401,400)=1$  so it may
         be constructed over $GF(3^{20}), GF(7^4), GF(401)$ and
         others.
         \item Now $GF(401)$ is a prime field and arithmetic therein
         is modular arithmetic; it is in fact the smallest field over
         which the Fourier $400\ti 400$ can be constructed. 

        \item  An Hermitian dual-containing code may be obtained  by
         working over $GF(p^{2l})$ when there exists an element of the
         required order in $GF(p^l)$. Just
         define $A$ as above to be  a Fourier $400\ti 400$ matrix over say
         $GF(401^2)$ using a $400^{th}$ root of unity in
         $GF(401^2)$. \item The $e_i$ in this case satisfy
         $e_i^{401}=e_{i*401}=e_i$ as $401\equiv 1 \mod 400$. Thus the
         code obtained is an Hermitian dual-containing $[400,350,51]$
         code from which a QECC $[[400,300,51]$ MDS code is
           designed.
\end{itemize} 
{\em            By taking `only' $\{e_0,\ldots, e_{349}\}$ the full
  power of all the
           rows of the  Fourier matrix is not utilised.}
           \begin{enumerate}
             \item Define $A$ as before
           and $B$ to be the matrix whose last $50$ rows are $\{e_{350},
           \ldots, e_{399}\}$ and whose first $350$ rows are zero
           vectors.\item Define $G[z] = A + Bz$. The convolutional code
           generated by $G[z]$ is a $(400,350,50;1)$ code. It may be
           shown to be non-catastrophic by writing down the right inverse
           of $G[z]$. \item
           The GSB for such a code is
           $(450-350)(\floor{\frac{50}{350}}+1)+50+1 = 101$. Using
           Lemma \ref{weight} it may be shown that the free distance
           of the code generated by $G[z]$ is $101$ so it's an MDS
           convolutional $(400,350,50;1,101)$ code. The distance is
           twice less $1$ of the distance of an MDS $[400,350,51]$
           linear block code.
\end{enumerate}
           To get an LCD linear block code of rate $\geq \frac{7}{8}$ and
           $d\geq 51$ it is necessary to go to
           length $401$ or higher. Use the methods of \cite{hurleylcd}.
           \begin{itemize} \item For length $401$, let $A$ be the
           matrix generated $\{e_0,e_1,e_{400}, e_2,e_{399}, \ldots,
           e_{175},e_{226}\}$. (The selection includes pairs
           $e_i,e_{401-i}$.) \item Then the code generated by $A$ is an LCD
           $[401,351,51]$ MDS code.

\item       The fields required for $n=401$ are fairly
           large. Go to $511= 2^9-1$ as here 
           $\ord(2,511)= 9$ so the field $GF(2^9)$ works and has
           characteristic $2$. \item Require $r\geq 511*\frac{7}{8}$ for a rate
           $\geq \frac{7}{8}$. Thus require $r\geq 448$. Take $r=449$ for
           reasons which will appear later. \item Let $F_{511}$ be  the Fourier
           $511\ti 511$ matrix over $GF(2^9)$ or for the Hermitian
           case over $GF(2^{18})$. \item Let $A$ be the matrix with 
           rows $\{e_0,e_1,\ldots,e_{448}\}$. \item Then $A$ is an
           $[511,449,63]$ MDS DC code -- and DC
           Hermitian code over $GF(2^{18})$. \item From this QECC MDS
           codes $[[511,387,63]]$  are designed and are Hermitian
           over $GF(2^{18})$.
\end{itemize}
           %To design an LCD code let $A$ be the matrix with rows $e_0,e_
{Convolutional}
  \begin{enumerate} \item        Let $B$ be the matrix of size $449\ti 511$ with last
           $62$ rows consisting of $\{e_{449}, \ldots, e_{510}\}$ and
         other rows consisting of zero vectors. \item Define
           $G[z] = A + Bz$. Then the code generated by $G[z]$ is a
           non-catastrophic MDS $(511,449,62:1,125)$ convolutional
           code. The proof of the distance follows the lines of Lemma
           \ref{weight}. It has twice the distance less $1$ of the
           corresponding MDS block linear $[511,449,63]$ code.

  \end{enumerate} Now design  LCD codes in $GF(2^9)$.
  \begin{enumerate} \item Let $A$ be the matrix with rows
           $\{e_0,e_1,e_{510}, e_2,e_{509}, \ldots, e_{224},
           e_{287}\}$. Notice the rows are in sequence and so $A$
           generates an $[511,449,63]$ linear block code. 
         \item The check matrix is $H\T= (f_{286},f_{225},f_{285},f_{226}, \ldots,
         f_{256},f_{255})$. \item Then $H$ consists of rows
         $\{e_{225},e_{286}, \ldots, e_{255},e_{256}\}$. Thus the code
         generated by $A$ has trivial intersection with the code
         generated by $H$ and so the code is an MDS LCD block linear
         $[511,449,63]$ code. This is an Hermitian LCD MDS code
         over $GF(2^{18})$.
\item 
         Let $B$ be the matrix whose last $62$ rows are
         $\{e_{286},e_{225},e_{285},e_{226}, \ldots, e_{256},e_{255}\}$
         and   whose first $449$ rows consists of zero vectors. \item Define
         $G[z] = A + Bz$. A check matrix for the code generated by $G[z]$
         is $(f_{286},f_{225},f_{285},f_{226}, \ldots,
         f_{256},f_{255}) - (0,0,\ldots,0,f_{254},f_{257}, \ldots,
         f_{224},f_{287})z = C+Dz$, say. \item  Recall we are in
         characteristic $2$. Now $C^T+D^Tz^{-1}
         = \begin{ssmatrix}e_{225}\\e_{286}\\ e_{226}\\e_{285}\\ \vdots
           \\ \vdots \\e_{255}\\e_{256}  \end{ssmatrix} + \begin{ssmatrix}
           0\\ \vdots\\0\\e_{257}\\e_{254}\\\vdots
           \\e_{287}\\e_{224}\end{ssmatrix}z^{-1}$. \item Thus the dual
         matrix is $\begin{ssmatrix}
           0\\\vdots\\0\\e_{257}\\e_{254}\\\vdots
           \\e_{287}\\e_{224}\end{ssmatrix}+ \begin{ssmatrix}e_{225}\\e_{286}\\ e_{226}\\e_{285}\\ \vdots
           \\ \vdots \\e_{255}\\e_{256} \end{ssmatrix}z$ 
         = $E+Fz$ say. \item It is relatively  easy to show that there is a matrix $K$
         such that $K(A+Bz) = E+Fz$ and so the code generated by
         $A+Bz$ is a DC convolutional 
         $(511,449,62;1)$ code. \item The GSB for such a  code is $125$ and
         this is the distance attained, so the code is a
         DC MDS convolutional $(511,449,62;1,125)$
         code. From this a quantum convolutional code may be
         designed. To obtain Hermitian DC, work in
         $GF(2^{18})$.
\end{enumerate}

         This can be extended to higher degrees. % $3$,$7$. %% finishing when
         %% running out of different rows of Fourier to be used.
         } 
          \end{example}  
         Thus in a sense:

         {\em DC block linear $\longrightarrow$ LCD  convolutional
           degree $1$ at twice the distance. }

         {\em LCD block linear $\longrightarrow$ DC convolutional
           degree $1$ at twice the distance.}

         The LCD block linear to give DC convolutional requires characteristic
         $2$. 
\section{Algorithms}\label{thms}
%% The rows of an $n\ti n$ matrix $F_n$ formed are
%% denoted in order by $\{e_0,e_1,\ldots,e_{n-1}\}$ and $n$ times the
%% columns of the inverse of $F_n$ in order are denoted by
%% $\{f_0,f_1,\ldots,f_{n-1}\}$. Thus 
%% $e_i=(1,\om^i,\om^{2i},\ldots,\om^{(n-1)i})$,
%% $f_i\T=e_{n-1},e_i\T=f_{n-i}, e_if_j=\de_{ij}$. Indices should be
%% taken modulo $n$. Now $n$ has an inverse in the field as the
%% Fourier matrix  exists in that field. 

%% Note also $e_i^l=e_{il}$ where  indices are taken modulo $n$.

%Summary of algorithms: 
\begin{Algorithm}\label{alg1}
  Construct block linear codes of rate $\geq R$ and distance $\geq (2t
  + 1)$ for $0<R<1$,  and positive integer $t$,  and with efficient decoding algorithm.

  {\em  A $[n,r,n-r+1]$ linear block code
    will be designed. Thus $\frac{r}{n} \geq R$, $(n-r+1)
    \geq 2t +1$. This requires $n-r\geq 2t$, $n(1-R)\geq 2t$ and so
    require $n\geq \frac{2t}{1-R}$.

    \begin{enumerate} \item
      Choose $n\geq \frac{2t}{1-R}$ and construct the Fourier $n\ti n$
    matrix $F_n$ over a suitable field. Now choose $ r\geq
    nR$.

\item  Select any $r$ rows of $F_n$ in arithmetic sequence with
    arithmetic difference $k$ satisfying $\gcd(k,n)=1$ and form the matrix
  $A$   consisting of these rows.

\item  The block linear code with generator matrix $A$ is an MDS
  $[n,r,n-r+1]$ code.
  \item The rate is $\frac{r}{n} \geq R$, and the
    distance $d= n-r+1 = n(1-R)+1 \geq 2t + 1$ as required.
    \end{enumerate} }
    \end{Algorithm}
\begin{Algorithm}\label{alg2}
  Construct block linear codes of rate $\geq R$ and distance $\geq (2t
  + 1)$ for given rational  $R$, $0<R<1$, and positive integer $t$,
  with efficient decoding algorithms,  over fields
  of characteristic $p$.

{\em    A $[n,r,n-r+1]$ code is designed in characteristic $p$. As in Algorithm \ref{alg1} it is
   required that $n\geq \frac{2t}{1-R}$. Require in addition that
   $\gcd(p,n) =1 $.
   
\begin{enumerate}\item Require $n\geq \frac{2t}{1-R}$ and  $\gcd(p,n)
  =1 $.
\item  Construct the Fourier $n\ti n$ matrix over a field of characteristic
 $p$.

\item  Proceed as in items 1-4 of Algorithm \ref{alg1}. \end{enumerate}}
   
    \end{Algorithm}

\begin{Algorithm}\label{alg3}
  Construct block linear codes DC codes of rate $\geq R$ with $
  \frac{1}{2}< R<1 $ and distance $\geq (2t
  + 1)$ for positive integer $t$,    with efficient decoding algorithm.

  {\em  A $[n,r,n-r+1]$ block dual-containing code
    is required. As before 
    require $n\geq \frac{2t}{1-R}$.

    \begin{enumerate} \item
      Choose $n\geq \frac{2t}{1-R}$ and construct the Fourier $n\ti n$
      matrix $F_n$ over a suitable field.

      \item Choose $r\geq nR$. As
    $R>\frac{1}{2}$ then $r>\floor{\frac{n}{2}} + 1$.

\item  Select $A$ to be the first $r$ rows of $F_n$, that is $A$
  consists of rows $\{e_0, e_1,\ldots, e_{r-1}\}$.
  Then $B\T=(f_r,\ldots, f_{n-1})$ satisfies $AB\T=0$. Now
  $B= \begin{ssmatrix}e_{n-r} \\ e_{n-r+1}\\ \vdots
    \\ e_1\end{ssmatrix}$. Thus the code generated by $A$ is a
    dual-containing MDS $[n,r,n-r+1]$ code. 

  \item The rate is $\frac{r}{n} \geq R$, and the
    distance $d= n-r+1 = n(1-R)+1 \geq (2t + 1)$ as required.
    \end{enumerate} } 
    \end{Algorithm}

\begin{Algorithm}\label{alg4}
  Design block linear DC codes of rate $\geq R>
  \frac{1}{2} $ and distance $\geq (2t
  + 1)$ for given $R, t$,  with efficient decoding algorithm, over
  fields of characteristic $p$.

  {\em  A $[n,r,n-r+1]$ block dual containing code
    is required. As before 
    require $n\geq \frac{2t}{1-R}$.

 \begin{enumerate} \item  Choose $n\geq \frac{2t}{1-R}$ and also such
   that $\gcd(p,n)=1$. \item  Construct the Fourier $n\ti n$
   matrix $F_n$ over a field of characteristic $p$.

\item   Now proceed as in items 1-4 of Algorithm \ref{alg3}. \end{enumerate}}
   \end{Algorithm}
\begin{Algorithm}\label{alg51} (i) Design MDS QECCs of form $[[n,2r-n,n-r+1]]$.
  
(ii) Design MDS QECCs of form $[[n,2r-n,n-r+1]]$ over a field of
  characteristic $p$.
  
 {\em  Method:

  (i) By Algorithm \ref{alg3} construct MDS DC codes
  $[n,r,n-r+1]$. Then by CSS construction, construct the MDS 
  $[[n,2r-n,n-r+1]]$ QECC. 

    (ii) By Algorithm \ref{alg4} construct MDS DC codes
  $[n,r,n-r+1]$ over a field of characteristic $p$. Then by CSS construction, construct the MDS
    $[[n,2r-n,n-r+1]$ QECC in a field of characteristic $p$.}

\end{Algorithm}

\begin{Algorithm}\label{alg5} Design LCD MDS codes of rate
  $\geq R$ and distance $\geq 2t+1$.

  {\em   This design follows from  \cite{hurleylcd}.

    \begin{enumerate}
\item  Choose $n\geq \frac{2t}{1-R}$ and $r\geq nR$.  
  
\item For such $n$ let $F_n$ be a Fourier $n\ti n$ matrix 
  with rows $\{e_0,\ldots, e_{n-1}\}$ in order and $n$ times the
  columns of the inverse in order are denoted by $\{f_0,\ldots,f_{n-1}\}$.
\item For $2r+1\geq nR$ and $r\leq \floor{\frac{n}{2}}$ define $A$ as
  follows. $A$ consists of row $e_0$ and  
  rows $\{e_1,e_{n-1},e_2,e_{n-2}, \ldots, e_r,e_{n-r}\}$ for . ($A$ consists of $e_0$ and pairs
  $\{e_i,e_{n-i}\}$ starting with $e_1,e_{n-1}$.)
  \item Set $B\T=(f_{r+1},f_{n-r-1}, \ldots, f_{\frac{n-1}{2}},f_{\frac{n+1}{2}})$
  when $n$ is odd and $B\T=(f_{r+1},f_{n-r-1}, \ldots,
  f_{\frac{n}{2}-1},f_{\frac{n}{2}+1}, f_{\frac{n}{2}})$ when $n$ is even.
  \item Then $AB\T=0$ and
  $B$ generates the dual code $\C^\perp$ of the code $\C$ generated by
    $A$.
    \item Using 
      $f_i\T=e_{n-i}$ it is easy to check that $\C\cap \C^\perp = 0$.
      \item Now
  the rows of $A$ are in sequence $\{n-r, n-r+1, \ldots, n-1, 0, 1,
  \ldots, r-1\}$ and so $A$ generates an MDS, LCD, $[n,2r+1, n-2r]$
  code.  \end{enumerate} }
  
\end{Algorithm}

      \begin{Algorithm}\label{primedual} Construct MDS and MDS
        DC codes and 
      MDS QECCs, all  over prime fields. Construct Hermitian such
      codes over $GF(p^2)$ for a prime $p$.

\begin{itemize}\item  Let $GF(p)=\Z_p$ be a prime field. This has a
primitive element of order $(p-1)=n$, say. \item Construct the Fourier
$n\ti n$ matrix over $GF(p)$.
\item Choose  $r > \frac{p-1}{2}$.
\item Construct $A=\begin{ssmatrix} e_0\\e_1\\\vdots
\\ e_{r-1}\end{ssmatrix}$.
\item $H\T= (f_{r},f_{r+1},\ldots, f_{n-1})$ is a check matrix.
\item $H$ consists of rows $\{e_{n-r}, e_{n-r-1}, \ldots, e_1\}$ and so
the code generated by $A$ is a DC $[n,r,n-r+1]$ MDS code  in $GF(p)$.
\item Over $GF(p^2)$ with $\langle u,v \rangle_H = \langle u, v^p
\rangle_E$ (and different $e_i$) the code generated by  $A$
is a Hermitian DC $[n,r,n-r+1]$ code.
\item  Construct $\begin{ssmatrix} e_0\\e_1\\ \vdots
\\ \vdots\\ e_{r-1}\end{ssmatrix} + \begin{ssmatrix} 0 \\ \vdots \\ 0 \\e_{r+1}
\\\vdots \\ e_{n-1}\end{ssmatrix}z$ where the second matrix  has
$(2r-n)$ initial zero rows. This gives  a convolutional
$(n,r,n-r;1,2(n-r)+1)$ code over $GF(p)$. This code is not
dual-containing even over $GF(p^2)$. 
\end{itemize}\end{Algorithm}

\begin{Algorithm}\label{dualratehalf}{Design infinite series of DC codes
    $[n_i,r_i,d_i]$,  
(i) with rates and rdists satisfying  $\lim_{i\rightarrow
      \infty}\frac{r_i}{n_i} = \frac{1}{2}$ and
    $\lim_{i\rightarrow\infty}\frac{d_i}{n_i} = \frac{1}{2}$, 
(ii) with rates and rdists satisfying  $\lim_{i\rightarrow \infty}\frac{r_i}{n_i} = \frac{s}{q}$ and $\lim_{i\rightarrow\infty}\frac{d_i}{n_i} = 1-\frac{s}{q}$.}
 
{\em \item{1:} Consider a series of $n_1< n_2 < n_3 < \ldots$.
and let $p_i$ be a prime such that $p_i\not|n_i$. Then
$\ord(p_i,n_i)= s_i$, for some $s_i$.
\item{2:} Construct the Fourier $n_i\ti n_i$ matrix
over $GF(p_i^{s_i})$. Let $r_i=\floor{\frac{n_i}{2}}+1$. Then 
$A=\begin{ssmatrix}e_0 \\e_1 \\ \vdots \\e_{r_i-1}\end{ssmatrix}$
generates  a
DC  $[n_i,r_i,n_i-r_i+1]$ MDS code over $GF(p_i^{s_i})$.
\item{3:} Construct the Fourier $n_i\ti n_i$ matrix
over $GF(p_i^{2s_i})$. Let $r_i=\floor{\frac{n_i}{2}}+1$. Then 
$A=\begin{ssmatrix}e_0 \\e_1 \\ \vdots \\e_{r-1}\end{ssmatrix}$
generates  an
Hermitian DC $[n_i,r_i,n_i-r_i+1]$ MDS code over $GF(p_i^{2s_i})$.
\item{4:} $\lim_{i\rightarrow \infty}\frac{r_i}{n_i} = \frac{1}{2},
\lim_{i\rightarrow \infty}\frac{n_i-r_i+1}{n_i}=\frac{1}{2}$.}
\end{Algorithm}
\begin{Algorithm}\label{ratehalfcharp}{Design infinite series of DC  codes
$[n_i,r_i,d_i]$ over fields of given characteristic $p$ with (i) rates and rdists satisfying  $\lim_{i\rightarrow \infty}\frac{r_i}{n_i} = \frac{1}{2}$,  $\lim_{i\rightarrow\infty}\frac{d_i}{n_i} =
\frac{1}{2}$, (ii) rates and rdists satisfying  $\lim_{i\rightarrow
  \infty}\frac{r_i}{n_i} = \frac{s}{q}> \frac{1}{2}$,  $\lim_{i\rightarrow\infty}\frac{d_i}{n_i} =1-
\frac{s}{q}$.}

\item{1:} Consider a series of $n_1< n_2 < n_3 < \ldots$ where
  $\gcd(p,n_i)=1$. Then  $\ord(p,n_i)= s_i$ for some $s_i$. 
\item{2:}\label{3} Construct the Fourier $n_i\ti n_i$ matrix
over $GF(p^{s_i})$. Let $r_i=\floor{\frac{n_i}{2}}+1$. Then 
$A=\begin{ssmatrix}e_0 \\e_1 \\ \vdots \\e_{r-1}\end{ssmatrix}$ is a
DC $[n_i,r_i,n_i-r_i+1]$ code over $GF(p^{s_i})$.
\item{3:} Construct the Fourier $n_i\ti n_i$ matrix
over $GF(p^{2s_i})$. Let $r_i=\floor{\frac{n_i}{2}}+1$. Then 
$A=\begin{ssmatrix}e_0 \\e_1 \\ \vdots \\e_{r_i-1}\end{ssmatrix}$ is a
Hermitian DC  $[n_i,r_i,n_i-r_i]$ code over $GF(p^{2s_i})$.
\item{4:} $\lim_{i\rightarrow \infty}\frac{r_i}{n_i} = \frac{1}{2},
\lim_{i\rightarrow \infty}\frac{n_i-r_i+1}{n_i}=\frac{1}{2}$.

(ii) The general fraction $R=\frac{p}{q}$ may  be obtained by choosing
$r_i=\floor{\frac{s}{q}n_i}$ in item 2. 

\end{Algorithm}

%% By choosing $r_i=\floor{\frac{3n_i}{4}}$ at Algorithm  \ref{3} above, infinite
%% series are obtained where $\lim_{i\rightarrow \infty}\frac{r_i}{n_i}=
%% \frac{3}{4}, \lim_{i\rightarrow \infty}\frac{d_i}{n_i} =
%% \frac{1}{4}$ where $d_i$ is the distance.

By taking a series of odd integers $2n_1+1 < 2n_2+1 < \ldots $ infinite
such series are obtained over fields of characteristic $2$. 
%% Of particular note are dual-containing codes over fields of
%% characteristic $2$.
The odd series $(2^2-1) < (2^3 -1) < (2^4-1)  < \ldots $ is particularly noteworthy.
Note $\ord(2,n)=s$ for $2^s-1=n$ and
$GF(2^s)$ contains an element of order $2^s-1$ describing {\em all} the
non-zero elements of $GF(2^s)$. The Fourier matrix of size $n\ti n$
may then be constructed over $GF(2^s)$. This has many nice consequences
for designing codes of different types in characteristic $2$.

This is illustrated as follows for the case $GF(2^5)$.
\begin{example}

Construct codes and particular types of codes over $GF(2^5)$ and Hermitian
such codes over $GF(2^{10})$.

{\em Let $F_{15}=F$ be the Fourier matrix of size $15\ti 15$ over
$GF(2^5)$ or as relevant over $GF(2^{10})$. 

\item{1:}  For $r>7$, $A=\begin{ssmatrix}e_0\\e_1\\ \vdots
\\e_{r-1}\end{ssmatrix}$ generates a DC MDS code
$[15,r,15-r+1]$.
\item{2:} For $r>7$, $A=\begin{ssmatrix}e_0\\e_1\\ \vdots
\\e_{r-1}\end{ssmatrix}$ considered as a matrix over $GF(2^{10})$
generates an Hermitian  DC MDS code $[15,r,15-r+1]$.
\item{3:} $G[z]= \begin{ssmatrix}e_0\\e_1\\e_2\\e_3\\e_4\\e_5\\e_6\\e_7
\end{ssmatrix} + \begin{ssmatrix}0\\e_8\\e_9\\e_{10}\\e_{11}\\e_{12}\\e_{13}
\\e_{14}\end{ssmatrix}z$ generates an MDS convolutional code
$(15,8,7;1,15)$. This has distance twice the distance less $1$ of the linear
MDS code $[15,8,8]$. But also this code is an LCD code. 
\item{4:} $A= \begin{ssmatrix}
e_0\\e_1\\e_{14}\\e_2\\e_{13}\\e_3\\e_{12}\\e_4\\e_{11}\end{ssmatrix}$
generates a LCD, MDS code $[15,9,7]$; see
\cite{hurleylcd} for details. $A*(f_5,f_{10},f_6,f_9,f_7,f_8)=0$.
\item{5:} Consider $A$ as above constructed in $GF(2^{10})$. Then $A$ is a 
Hermitian LCD code over $GF(2^{10})$.
\item{6:} $A= \begin{ssmatrix}
e_0\\e_1\\e_{14}\\e_2\\e_{13}\\e_3\\e_{12}\\e_4\\e_{11}\end{ssmatrix}$,
$B=\begin{ssmatrix}0\\0\\0\\e_5\\e_{10}\\e_6\\e_9\\e_7\\e_8\end{ssmatrix}$
designs  the MDS convolutional code $(15,9,6;1,13)$, generated by $G[z]=
A+Bz$. This code in addition is a {\em dual-containing} MDS convolutional code. 
\item{7:} By taking $r =\floor{\frac{3}{4}15}+1 = 12$ or $r =\floor{\frac{3}{4}15} = 11$ codes of rate
`near' $\frac{3}{4}$ are obtained; in this case codes of rate $\frac{12}{15}
=\frac{4}{5}$ or $\frac{11}{15}$ attaining the MDS are obtained.}  
\end{example}
\begin{Algorithm}\label{herchar2} Construct infinite series MDS codes of various types
  of linear block and convolutional codes 
over fields of the form $GF(2^i)$ and Hermitian such codes over fields
of the form $GF(2^{2i})$. 

\item{1.} First note  $\ord(2,2^i-1)=i$ for given $i$. 
\item{2.} Construct the
$n_i\ti n_i$ Fourier matrix $F_{n_i}$ over $GF(2^i)$ with $n_i=2^i-1$.
\item{3.} For $r_i>\frac{n_i}{2}$ let $\mathcal{C}_{r_i}$ be the code generated by rows
$\langle e_0,e_1,\ldots, e_{r-1}\rangle$. Then $\mathcal{C}_{r_i}$ is an
$[n_i,r_i,n_i-r_i+1]$ dual-containing MDS linear code over $GF(2^i)$.
\item This gives an infinite series of codes $\mathcal{C}_{r_i}$ of type $[n_i,r_i,n_i-r_i+1]$.
\item{4.} For $r_i=\frac{n_i}{2}$ this gives an infinite series
$\mathcal{C}_{r_i}$ with $\lim_{i\rightarrow \infty}\frac{r_i}{n_i} =
\frac{1}{2}$ and $\lim_{i\rightarrow \infty}\frac{d_i}{n_i} =
\frac{1}{2}$ where $d_i=n_i-r_i+1$ is the distance.
 \item{5:} For $r_i=\floor{\frac{3n_i}{4}}$ this gives an infinite series
$\mathcal{C}_{r_i}$ with $\lim_{i\rightarrow \infty}\frac{r_i}{n_i} =
\frac{3}{4}$ and $\lim_{i\rightarrow \infty}\frac{d_i}{n_i} =
\frac{1}{4}$ where $d_i=n_i-r_i+1$ is the distance.
\item{6:} Other such
infinite series can be obtained with different fractions by letting
$r_i=\floor{Rn_i}$ ($R$ rational) and then series $\mathcal{C}_{r_i}$ with $\lim_{i\rightarrow \infty}\frac{r_i}{n_i} =
R$ and $\lim_{i\rightarrow \infty}\frac{d_i}{n_i} =
1-R$ are obtained where $d_i=n_i-r_i+1$ is the distance.
 
\item Construct the 
$n_i\ti n_i$ Fourier matrix $F_{n_i}$ over $GF(2^{2i})$ with $n_i=2^{i}-1$.
For $r_i>\frac{n_i}{2}$ let $\mathcal{C}_{r_i}$ be the code generated by rows
$\{ e_0,e_1,\ldots, e_{r_i-1}\}$. Then $\mathcal{C}_{r_i}$ is an
$[n_i,r_i,n_i-r_i+1]$ DC MDS Hermitian linear code over
$GF(2^{2i})$. 
\item This gives an infinite series of $\mathcal{C}_{r_i}$ of type
$[n_i,r_i,n_i-r_i+1]$ which are Hermitian DC.
\item{} For $r_i=\frac{n_i}{2}$ this gives an infinite series of
Hermitian DC 
$\mathcal{C}_{r_i}$ codes  with $\lim_{i\rightarrow \infty}\frac{r_i}{n_i} =
\frac{1}{2}$ and $\lim_{i\rightarrow \infty}\frac{d_i}{n_i} =
\frac{1}{2}$ where $d_i=n_i-r_i+1$ is the distance.
 \item $r_i=\floor{\frac{3n_i}{4}}$ designs an infinite series
of Hermitian DC codes $\mathcal{C}_{r_i}$ with $\lim_{i\rightarrow
\infty}\frac{r_i}{n_i} = \frac{3}{4}$ and $\lim_{i\rightarrow
\infty}\frac{d_i}{n_i} = \frac{1}{4}$ where $d_i=n_i-r_i+1$ is the
distance.
\item Other such infinite series of Hermitian codes
  can be obtained with different fractions by choosing
  $r_i=\frac{sn_i}{q}$ for a fraction $\frac{s}{q}<1$.  

\end{Algorithm}
%% \item{4.} For $r>\frac{n}{2}$ let $A$ be the code generated by rows
%% $\langle e_0,e_1,\ldots, e_{r-1}\rangle$. Then $A$ is an $[n,r,n-r+1]$
%% Hermitian dual-containing MDS linear code over $GF(2^{2i})$. 
\begin{Algorithm}\label{primetype} Construct infinite series of codes of various `types'
over prime fields $GF(p)$ and with  Hermitian inner product over fields $GF(p^2)$.

{\em Let $p_1 <p_2 < \ldots $ be an infinite series of primes.

For $p_i$ construct the Fourier $n_i\ti n_i$ matrix over $GF(p_i)$ where
$n_i=p_i-1$.

For $r_i>\frac{n_i}{2}$ let $\mathcal{C}_{r_i}$ be the code generated by rows
$\langle e_0,e_1,\ldots, e_{r_i-1}\rangle$. Then $\mathcal{C}_{r_i}$ is an
$[n_i,r_i,n_i-r_i+1]$ DC MDS linear code over
$GF(p_i)$. {\em The arithmetic is modular arithmetic.} 

This gives an infinite series of codes $\mathcal{C}_{r_i}$ of type $[n_i,r_i,n_i-r_i+1]$.

For $r_i=\frac{n_i}{2}+1$ this gives an infinite series
$\C_{r_i}$ of type $[n_i,r_i,n_i-r_i+1]$ DC codes over prime fields
with $\lim_{i\rightarrow \infty}\frac{r_i}{n_i} =
\frac{1}{2}$ and $\lim_{i\rightarrow \infty}\frac{d_i}{n_i} =
\frac{1}{2}$ where $d_i=n_i-r_i+1$ is the distance.

For $r_i=\floor{\frac{3n_i}{4}}$ this gives an infinite series of such 
$\mathcal{C}_{r_i}$ with $\lim_{i\rightarrow \infty}\frac{r_i}{n_i} =
\frac{3}{4}$ and $\lim_{i\rightarrow \infty}\frac{d_i}{n_i} =
\frac{1}{4}$ where $d_i=n_i-r_i+1$ is the distance.

%For $r_i=\floor{\frac{7n_i}{8}}$ ....

For $r_i=\floor{\frac{p}{q}n_i}$
infinite series are  obtained with $\lim_{i\rightarrow \infty}\frac{r_i}{n_i} =
\frac{p}{q}$ and $\lim_{i\rightarrow \infty}\frac{d_i}{n_i} =
1-\frac{p}{q}$ where $d_i=n_i-r_i+1$ is the distance.

Construct the 
$n_i\ti n_i$ Fourier matrix $F_{n_i}$ over $GF(p_i^{2})$.
For $r_i>\frac{n_i}{2}$ let $\mathcal{C}_{r_i}$ be the code generated by rows
$\{e_0,e_1,\ldots, e_{r_i-1}\}$. Then $\mathcal{C}_{r_i}$ is an
$[n_i,r_i,n_i-r_i+1]$ Hermitian DC MDS linear code over
$GF(p_i^2)$.

This gives an infinite series of $\mathcal{C}_{r_i}$ of type
$[n_i,r_i,n_i-r_i+1]$ which are Hermitian DC over
$GF(p_i^2)$. 

For $r_i=\floor{\frac{n_i}{2}}+1$ this gives an infinite series of
Hermitian DC  
$\mathcal{C}_{r_i}$ codes with $\lim_{i\rightarrow \infty}\frac{r_i}{n_i} =
\frac{1}{2}$ and $\lim_{i\rightarrow \infty}\frac{d_i}{n_i} =
\frac{1}{2}$ where $d_i=n_i-r_i+1$ is the distance.

For $r_i=\floor{\frac{3n_i}{4}}$ this gives an infinite series
of Hermitian DC $\mathcal{C}_{r_i}$ codes with $\lim_{i\rightarrow
\infty}\frac{r_i}{n_i} = \frac{3}{4}$ and $\lim_{i\rightarrow
\infty}\frac{d_i}{n_i} = \frac{1}{4}$ where $d_i=n_i-r_i+1$ is the
distance.

%For $r_i=\floor{\frac{7n_i}{8}}$ \ldots

For $r_i=\floor{\frac{p}{q}n_i}$ an 
infinite series are  obtained with $\lim_{i\rightarrow \infty}\frac{r_i}{n_i} =
\frac{p}{q}$ and $\lim_{i\rightarrow \infty}\frac{d_i}{n_i} =
1-\frac{p}{q}$ where $d_i$  is the distance.

%Other such infinite series can be obtained with different fractions. 
}
\end{Algorithm}

     The following algorithm explains how to design convolutional
     codes with order twice the distance of the corresponding MDS code of
     the same length and rate.
\begin{Algorithm}\label{algor3} Design convolutional MDS memory $1$
  codes to the order of twice the distance of the linear block MDS
  codes of the same length and rate.
  {\em 
  \begin{enumerate} \item Let $F_n$ be a Fourier $n\times n$ matrix. Denote the
    rows in order of $F_n$ by $\{e_0,e_1,\ldots,e_{n-1}\}$ and $n$
    times the columns
    of the inverse of $F_n$ by $\{f_0,f_1,\ldots, f_{n-1}\}$. Then
    $e_if_j=\delta_{ij},f_i^T = e_{n-i}, e_i^T=f_{n-i}$ with indices
    taken $\mod n$.
\item\label{ret} Let $r > \floor{\frac{n}{2}}$. Let $A$ be the matrix
  with first 
  $r$ rows $\{e_0,e_1,\ldots,e_{r-1}\}$ of $F_n$  and let $B$ be the
  matrix whose last rows are $\{e_r,\ldots, e_{n-1}\}$ in order and whose first
  $(n-r)$ rows consists of zero vectors.
\item Define $G[z]=A+Bz$.
  Then $G[z]$ is a generating matrix for a non-catastrophic convolutional
  code $(n,r,n-r;1,2(n-r)+1)$ of free distance $2(n-r)+ 1$. A control 
  matrix is easily written down, as is a right inverse for
  $G[z]$. 
  \end{enumerate} }
\end{Algorithm}
  
  Thus the MDS convolutional code produced of rate $\frac{r}{n}$ has twice
  the distance, 
  less $1$, of the MDS linear code $[n,r,n-r+1]$ with the same
  rate and length. %and is an MDS convolutional code. 

Note that $r$ can be any integer $>\floor{\frac{n}{2}}$ so all rates
$\frac{r}{n}$ for $ \floor{\frac{n}{2}}<r<n$ are obtainable. The dual
code has rate $(1-R)$ where $R=\frac{r}{n}$ is the rate of $\C$ so 
rates $R$ with $R<\frac{1}{2}$ are obtainable 
  
  Alternatively item 2.  of Algorithm \ref{algor3} may be replaced by
  taking rows in arithmetic sequence as follows: \\
  Let $r > \floor{\frac{n}{2}}$ and  $A$ be formed from $F_n$ by
    taking $r$ rows in geometric sequence with geometric
    difference $k$ satisfying $\gcd(k,n)=1$. Define $B$ to be the  
  matrix whose last rows are the other rows of $F_n$ not in $A$  (which also
  are in geometric sequence satisfying the  $\gcd$ condition)  and whose first
  $(n-r)$ rows consist of zero vectors.

The methods are illustrated in the following examples.

%% \begin{example} Consider the $F_{10}$ Fourier matrix.

%%   {\em Then
%%         $A=\langle e_0,e_1,e_2,e_3,e_4,e_5 \rangle$ generates an
%%         $[10,6,5]$ MDS linear code.

%%   Construct 
%%         $$G[z]
%%         = \begin{ssmatrix}e_0\\e_1\\e_2\\e_3\\e_4\\e_5 \end{ssmatrix}
%%         + \begin{ssmatrix}0\\0\\e_6\\e_7\\e_8\\e_9\end{ssmatrix}z = A+Bz$$

%%           A convolutional  $(10,6,4;1, d_f)$ code is obtained.
%%           Now $(A+Bz)*((f_6,f_7,f_8,f_9) - (f_1,f_2,f_3,f_4)z) = 0$
%%           and $(A+Bz)*(f_0,f_1,f_2,f_3,f_4,f_5) = 10I_6$. The code
%%           matrix has a right inverse and hence  the code is
%%           non-catastrophic. The GSB for such a code is
%%           $(10-6)(\floor{\frac{4}{6}}+1) + 4 + 1 = 4+4+1=9$. This is
%%           actually the free distance of the code which almost doubles
%%           the distance of $A$. }
%% \end{example} 

\quad

The cases $2^8=256$ and the near prime $257$ are worth noting.
\begin{example} (i) Construct DC MDS codes of length
  $255$ of various permissible rates over $GF(2^8)$.
  (ii) Construct Hermitian DC MDS codes of length $255$
  of various permissible rates over $GF(2^{16})$.
  (iii) Construct QECC MDS codes of length
  $255$ of various rates over $GF(2^8)$.
  (iv) Construct Hermitian QECC MDS codes of length $255$
  of various rates over $GF(2^{16})$.
  (v) Construct LCD MDS codes of length
  $255$ of various rates over $GF(2^8)$.
  (vi) Construct Hermitian LCD MDS codes of length $255$
  of various rates over $GF(255^2)$.
  
  {\em 

    \begin{enumerate}
      \item Over $GF(2^8)$ construct the Fourier $255\ti 255$
matrix $F$. 

\item For $r> \floor{\frac{255}{2}} = 127$ let $A$ be the code generated
by the first $r$ rows of $F$. Then $A$ is a DC 
$[255,r,255-r+1]$ code. 

\item For $r=128$ the DC $[255,128,128]$ code of rate about
$\frac{1}{2}$ and rdist  of about $\frac{1}{2}$ is designed. 

\item For $r= \floor{\frac{255*3}{4}}=191$ the code $[255,191,65]$ of
rate about $\frac{3}{4}$ and rdist of about $1/4$ is obtained.

\item For $r= \floor{\frac{255*7}{8}}=223$ the code $[255,223,33]$ of
rate about $\frac{7}{8}$ and rdist  of about $\frac{1}{8}$ is obtained. This
can correct $16$ errors.

\item Over $GF(2^{16})$ construct the Fourier $255\ti 255$ matrix. 

\item For $r> \floor{\frac{255}{2}} = 127$, let $A$ be the code generated
by the first $r$ rows of $F$. Then $A$ is an Hermitian  DC 
$[255,r,255-r+1]$ code. 

\item For $r=128$  the Hermitian DC $[255,128,128]$ code of rate about
$\frac{1}{2}$ and rdist of  about $\frac{1}{2}$ is designed. 

\item For $r= \floor{\frac{255*3}{4}}=191$ the Hermitian DC
code $[255,191,65]$ of rate about $\frac{3}{4}$ and rdist  of about
$\frac{1}{4}$ is obtained. 

\item For $r= \floor{\frac{255*7}{8}}=223$ the Hermitian DC 
code $[255,223,33]$ of rate about $\frac{7}{8}$ and rdist 
of about $\frac{1}{8}$ is obtained. This can correct $16$ errors.
    \end{enumerate}

    To obtain the  QECCs,  apply the CSS construction to the DC
    codes formed.

    \quad \quad
    
LCD  codes are designed as follows; see \cite{hurleylcd}
where the method is devised.
\begin{enumerate} \item Let $\C$ be the code generated by the rows $\langle
e_0,e_1,e_{254},e_2,e_{253}, e_3,e_{252}, \ldots,
e_{r},e_{255-r}\rangle$ for $2r<255$. \item Then $\C$ is a $[255,
  2r+1,255-2r]$ code;
notice that the rows of $A$ are in sequence $\{255-r, 255-r+1,\ldots 
254,0,1,2,\ldots, r\}$ so the code is MDS. \item
The dual code of $\C$ is the
code generated by the transpose of \\
$(f_{255-r-1}, f_{r+1}, f_{r+2},f_{255-r-2}, \ldots,
f_{r+2},f_{255-r-2})$ and this consists of rows \\
$\{e_{r+1},e_{255-r-1}, \ldots, e_{255-r-2},e_{r+2}\}$.  Thus $\C\cap
\C^\perp=0$ and so the code is an LCD MDS code.
\item To get LCD MDS Hermitian codes of length $255$ of various  rates
  as above, work in $GF(2^{16})$.  \end{enumerate} 
}
\end{example}

\begin{lemma}\label{lcd} Let $A,B,C,D$ be matrices of the same
  size $r \ti n$. Suppose the code generated by $A$ intersects trivially the
  code generated by $C$ and the code generated by $B$ intersects
  trivially the the code generated by $D$. Then the convolutional code
  generated by $A+Bz$ intersects trivially the convolutional code
  generated by $C\pm Dz$.
  \end{lemma}

\begin{proof} Compare coefficients of $P[z](A+Bz)$ with coefficients
  of $Q[z](C\pm Dz)$ for $1\ti r$ polynomial vectors $P[z],Q[z]$,
  $P[z] = P_0+P_1z+ \ldots , Q[z]=Q_0+Q_1z+ \ldots $. In turn get
  $P_0=0=Q_0$ then $P_1=0=Q_1$ and so on.
\end{proof}
 Using this, LCD MDS convolutional codes may be designed leading on
 from the DC codes designed in Algorithm \ref{alg4}.

 \begin{Algorithm}\label{convlcd} Design MDS LCD convolutional codes of the order of
   twice the distance of the MDS DC block code with the
   same length and rate.
{\em 
   First of all design the DC codes as in Algorithm
   \ref{algor3}.
     \begin{enumerate} \item Let $F_n$ be a Fourier $n\times n$
       matrix. %% Denote the 
    %% rows in order of $F_n$ by $e_0,e_1,\ldots,e_{n-1}$ and the columns
    %% of the inverse of $F_n$ by $f_0,f_1,\ldots, f_{n-1}$. Then
    %% $e_if_j=\delta_{ij},f_i^T = e_{n-i}, e_i^T=f_{n-i}$ with indices
    %% taken $\mod n$.
\item\label{ret1} Let $r > \floor{\frac{n}{2}}$. Define $A$ to be the matrix
  of the first 
  $r$ rows $(e_0,e_1,\ldots,e_{r-1})$ of $F_n$  and define  $B$ be the
  matrix whose last rows are $e_r,\ldots, e_{n-1}$ in order and whose first
  $n-r$ rows consist of zero vectors.
\item Define $G[z]=A+Bz$.
  Then $G[z]$ is a generating matrix for a non-catastrophic convolutional
  code $(n,r,n-r;1,2(n-r)+1)$ of free distance $2(n-r)+ 1$. A check
  matrix is easily written down, as is a right inverse for
  $G[z]$.
  \item The code generated by $G[z]$ is an LCD MDS convolutional
    code. This is shown as follows:

    A control matrix for the code is
    $H\T[z]=(f_r,\ldots,f_{n-1}) - (f_{n-r},f_{n-r+1},\ldots,
    f_{r-1})z$. Then $H[z^{-1}]= \begin{ssmatrix} e_{n-r} \\ e_{n-r-1} \\ \vdots
        \\ e_1 \end{ssmatrix} - \begin{ssmatrix}e_{r}\\e_{r-1}\\ \vdots
        \\ e_{n-r+1} \end{ssmatrix}z^{-1}$. Thus a control matrix is
      $\begin{ssmatrix} e_{n-r} \\ e_{n-r-1} \\ \vdots 
        \\ e_1 \end{ssmatrix}z - \begin{ssmatrix}e_{r}\\e_{r-1}\\ \vdots
        \\ e_{n-r+1} \end{ssmatrix}$ equal to say $-C+Dz$. Now the code
      generated by $C$ has trivial intersection with the code
      generated by $A$ and the code generated by $D$ has trivial
      intersection with the code generated by $B$. Hence by Lemma
      \ref{lcd} the convolutional code $\C$ generated by $A+Bz$ has trivial
      intersection with the code generated by $C-Dz$. Hence $\C$ is an
      LCD convolutional MDS $(n,r,n-r;1, 2(n-r)+1)$ code.   
  \end{enumerate} }
\end{Algorithm}

\subsection{QECC Hermitian} Of particular interest are QECCs with
Hermitian inner product. These need to be designed over fields $GF(q^2)$
where the Hermitian inner product is defined by ${\langle u, v
  \rangle}_H = {\langle u, v^q\rangle}_E$. Hermitian QECCs can be
designed  by the CSS 
construction from DC Hermitian  codes. A separate algorithm is given
below although it follows by similar methods to those already designed. 
\begin{Algorithm}\label{QECC} Construct QECCs over $GF(q^2)$.

  \item $GF(q^2)$ has an element of order $q^2-1$ and hence an element
    of order $q-1$ as $q^2-1=(q-1)(q+1)$.
    \item Let $\om$ be an element of order $(q-1)$ in $GF(q^2)$. Let
      $n=(p-1)$ and construct the Fourier $n\ti n$ matrix $F_n$ with this
      $\om$. The rows of $F_n$ are denoted by
      $\{e_0,e_1,\ldots,e_{n-1}\}$.
      \item Let $r > \floor{\frac{n}{2}}$ and define $A$ to be the
        matrix with rows $\{e_0,\ldots, e_{r-1}\}$. Then the code
        generated by $A$ is an Hermitian  DC MDS
        $[n,r,n-r+1]$ code over $GF(q^2)$.
        \item Use the CSS construction to form a QECC MDS Hermitian code
          $[[n,2r-n,n-r+1]$ code.
          \item For $r=\floor{\frac{n}{2}}+1 $ a DC
            MDS code of rate about $\frac{1}{2}$ is obtained and an
            MDS QECC 
            of rate $0$ and rdist of about $\frac{1}{2}$. For
            $r=\floor{\frac{3n}{4}}$ a 
            DC MDS code is obtained  of rate about
            $\frac{3}{4}$ and a 
            MDS, QECC of rate about $\frac{1}{2}$ and rdist of  about
            $\frac{3}{4}$. 
          \item Higher rates may be obtained.
\end{Algorithm}

Infinite series of such codes may also be obtained. The following
Algorithm gives an infinite series of characteristic $2$ such codes but other
characteristics are obtained similarly; the characteristics may be mixed.

\begin{Algorithm}\label{char2rate} Construct infinite series of
  characteristic $2$ 
  Hermitian DC $[n_i,r_i,n_i-r_i+1]$ codes $\C_i$
  in which (i) $\lim_{i\rightarrow \infty}\frac{r_i}{n_i} = R$, for $ 1>  R\geq \frac{1}{2}$ and (ii) $\lim_{i\rightarrow \infty}\frac{d_i}{n_i} = 1-
  R$; $R$ here is rational.  From this derive infinite series of Hermitian MDS QECCs $\D_i$ of form
  $[[n_i,2r_i-n_i, n_i-r_i+1]]$ in which $\lim_{i\rightarrow
    \infty}\frac{2r_i-n_i}{n_i} =  (2R-1), $ and (ii) $\lim_{i\rightarrow \infty}\frac{d_i}{n_i} = (1-R)$. 

  {\em Consider $GF(2^{2i})$. This has an element of order $(2^i-1)=n_i$ and
  use this to form the Fourier $n_i\ti n_i$ matrix over $GF(2^{2i})$. Let
  $r_{j,i}>\floor{\frac{n_i}{2}}$ and $A$ be the matrix with rows $\{e_0,\ldots,
  e_{r_{j,i}-1}\}$. Then the code $\C_{i,j}$ generated by $A$ is a Hermitian
  dual-containing code $[n_i,r_{j,i},n_i-r_{j,i}+1]$ code.
 This gives an infinite series $\C_{i,j}$ of Hermitian dual-containing 
 codes in characteristic $2$. The $r{i,j}$ can vary for each $GF(2^{2i})$.
 
  Now fix  $r_{j,i}=\floor{\frac{n_i}{2}} + 1=r_i$ for each $\C_{i,j}$ and
  let $C_i$ be the codes obtained. This  gives the infinite series $\C_i$ of
  $[n_i,r_i,n_i-r_i+1]$ codes and    
  $\lim_{i\rightarrow \infty}\frac{r_i}{n_i} =\frac{1}{2},
  \lim_{i\rightarrow \infty}\frac{d_i}{n_i} =\frac{1}{2}$.

  Fixing $r_{j,i}=\floor{\frac{3n_i}{4}}=r_i$ gives an infinite series
  $\C_i$ of  
$[n_i,r_i,n_i-r_i+1]$ codes with $\lim_{i\rightarrow \infty}\frac{r_i}{n_i} =\frac{3}{4},
  \lim_{i\rightarrow \infty}\frac{d_i}{n_i} =\frac{1}{4}$.

  Fixing $r_{j,i}=\floor{\frac{pn_i}{q}} = r_i$, $1/2 < p/q <1$ gives an
  infinite series $\C_i$ of $[n_i,r_i,n_i-r_i+1]$ codes 
and $\lim_{i\rightarrow \infty}\frac{r_i}{n_i} =\frac{p}{q},
\lim_{i\rightarrow \infty}\frac{d_i}{n_i} =1-\frac{p}{q}$.

The infinite series of Hermitian QECCs with limits as specified is
immediate. 

  }
  
\end{Algorithm}

\subsection{Higher memory}%%  Higher memory is briefly touched upon only
%% and needs to  be expanded.
Higher memory MDS convolutional codes may be 
obtained by  this general  method of using all the rows of an invertible
`good' matrix. %% such as a Fourier or 
%% Vandermonde matrix.
The principle is established in
\cite{hurleyconv} where rows of an invertible matrix are used to construct
convolutional codes. Here just one example is given and the general
construction is left for later work; some extremely nice codes are
obtainable by the method. %% The basic idea is to take rows of the invertible
%% matrix.
%%  and use them all but do not have the same rows at the same
%% level of the degrees. 

\begin{example} {\em Consider again, Example \ref{pro7},  the Fourier
    $7\ti 7$ matrix over $GF(2^3)$ with rows $\{e_0, \ldots, e_6\}$
    and $7$ times the columns of the inverse denoted by $\{f_0, \ldots, f_6\}$.

  Construct $G[z] = \begin{ssmatrix} e_0\\e_1\\e_2 \end{ssmatrix}
    + \begin{ssmatrix}0\\e_3\\e_4\end{ssmatrix}z + \begin{ssmatrix}
        e_5\\0\\e_6 \end{ssmatrix}z^2 $

        Then $G[z]$ is a convolutional code of type $(7,3,5;2)$; the
        degree is $5$. The GSB for such a code is
        $(7-3)(\floor{\frac{5}{3}+1} + 5+1 = 4*2 + 5+1 = 14$. In fact
        the free distance is actually $14$. This may be shown in an
        analogous way to the proof of Lemma \ref{weight}.

        $G[z]*(f_3,f_4,f_5,f_6) - (f_1,f_2,0,0) - (0,0,f_0,f_2)z^2)=
        0$, $G[z]*((f_0,f_1,f_2) = 7I_3$.

        The result is that the code generated by $G[z]$
        is a non-catastrophic convolutional
        MDS  $(7,3,5;2,14)$ code.

        Note the free distance
        attained is $5*3-1$ where $5$ is the free distance of a
        $[7,3,5]$ MDS code; the distance is tripled less $1$. This is
        a  general principle -- the free distance is of order three
        times the distance of the same length and dimension MDS code.

        It's not a dual-containing code nor a LCD code. To get such
        codes requires a compromise on the distance.  $G[z]
        = \begin{ssmatrix} e_0\\e_1\\e_2 \end{ssmatrix} 
    + \begin{ssmatrix}0\\e_3\\e_4\end{ssmatrix}z + \begin{ssmatrix}
        0\\e_5\\e_6 \end{ssmatrix}z^2 $. This give $(7,3,4;2)$
      convolutional code which turns out to be an LCD code but the
      free distance is only $7$. The GSB for such a code is $13$.} 
\end{example}
Convolutional codes $(n,r,\delta)$ have  maximum free
distance $(n-r)(\floor{\frac{\delta}{r}}+1) + \delta + 1$. When
$r>\delta$ this maximum free distance is $n-r+\delta + 1$.
Here is a design method for maximum free distance memory $1$
convolutional codes.

  \begin{example}

    Construct a convolutional code of rate
    $\frac{15}{16}$ which has free distance $\geq 61$.
    
    {\em     It is  required to construct  a $(n,r,n-r)$ convolutional
      code such that $\frac{r}{n}\geq
     \frac{15}{16}$ and $2(n-r)+1\geq 61$. Thus  require $(n-r)\geq 60$ and
     hence require $n(1-R)\geq 30$. Thus $n\geq 30/(1-R)\geq 30*16= 480$.
     Construct the Fourier $480 \times 480$ matrix over a suitable
     field. Require $\frac{r}{n}\geq \frac{15}{16}$ and $r\geq
     \frac{15}{16}*480 =450$. Now by
     Algorithm \ref{algor3} construct the $(480, 450, \delta) $
     convolutional code with $\delta = 30$. This code has free distance
     $2(n-r) + 1 = 60 + 1$ as required. The rate is $\frac{450}{480} =
     \frac{15}{16}$.
     
     The Fourier $480\times 480$ may be constructed
     over  a field of characteristic $p$ where $\gcd(p,480)=1$.
     Now $7^4 \equiv 1 \mod
     480$ so the field  $GF(7^4)$ can be used.
     This has an element of order $480$
     and the Fourier matrix of $480\times 480$ exists over $GF(7^4)$. 

     Suppose now a field of characteristic $2$ for example is
     required. Then replace ``$n\geq 480$'' by ``$n\geq 480$ and $\gcd(2,
     n)=1$ ''.  As we shall see,  it is convenient to take
     $n$ to be $(2^s -1)$ and in this case take $n=2^6-1 = 511$ in which
     case the arithmetic is done in $GF(2^6)$. }

The first prime greater than $480$ is $487$ so the construction can be
done over the prime field $GF(487)$.   
  \end{example}

  These codes have twice the error-correcting capability as MDS codes
  and of the same rate so should be very useful as codes.

  Series of block linear codes which are DC, QECCs, LDC
  are designed in the Algorithms \ref{alg1} to \ref{alg4}. Now we work on the
  types of convolutional codes that can be formed from these types
  when extending according to Algorithm \ref{algor3}.
  Thus design methods for convolutional DC,
  QECCs and LCD codes are required. 
  %% Below is shown how dual-containing codes may be derived in similar
  %% manner with similar error-correcting capabilities. In addition
  %% dual-containing codes lead to the design of quantum codes.
\subsubsection{Comment}
  From a recent article: 
  ``Far more efficient quantum error-correcting codes are needed to
  cope with the daunting error rates of real qubits. The effort to
  design better codes is “one of the major thrusts of the field,”
  Aaronson said, along with improving the hardware.

  Ahmed Almheiri, Xi Dong and Daniel Harlow \cite{alm} did calculations
  suggesting that this holographic “emergence” of space-time works
  just like a quantum error-correcting code. They conjectured in the
  Journal of High Energy Physics that space-time itself is a code — in
  anti-de Sitter (AdS) universes, at least. This lead to  a
  wave of activity in the quantum gravity community, leading to new
  new impulse to quantum  error-correcting codes that could capture more
  properties of space-time.
  
  What this is saying is that
  ``Ahmed Almheiri, Xi Dong and Daniel Harlow originated a powerful new
  idea that the fabric of space-time is a quantum error-correcting code''.

\subsection{DC LCD convolutional}
Convolutional DC codes over fields of characteristic
$2$ may be designed  as follows. %% leading on from the LCD MDS block linear codes as
%% already designed.
% in Algorithm \ref{alg3}. 
\begin{Algorithm}\label{algor1}
  
  \begin{enumerate}

  \item Let $F_{2m+1}$ be a Fourier matrix over a
  field of characteristic $2$. Denote its rows in order  by $\{e_0,
  e_1, \ldots,   e_{n-1}\}$ and the columns of its inverse times $n$
  is denoted by $\{f_0,f_1, \ldots, f_{n-1}\}$ in order, where $n=2m+1$. Then
  $e_if_j=\delta{ij}, f_i\T = e_{n-i}, e_i\T = f_{n-i}$.
  \item\label{jit} Choose the matrix $A$ as follows. Let $e_0$ be its first row
    and then choose $r$  pairs $\{e_i,e_{n-i}\}$ for the other rows
    and such  that $2r \geq m$. Thus $A$ has $(2r+1)$ rows and
    $A$ is an $(2r+1) \times n$ matrix. 
  \item Choose $B$ with first $(4r-2m+1) $ rows consisting of the zero
    vector and the other $2(m-r)$ rows
    consisting of the rest of the pairs ${e_i,e_{n-i}}$ ($m-r$ pairs)
    not used in 
    item \ref{jit}.  Then $B$ is  a $(2r+1) \times n$ matrix. 
  \item Construct $G[z]=A+Bz$. 
    \item $G[z]$ generates a convolutional dual-containing code from which
      quantum convolutional codes may be constructed. The control
      matrix of the code is easy to construct. There is a matrix $K$
      such that $GK=I_{2r+1}$ thus ensuring the code is
      non-catastrophic. The degree, $\delta$, of the code is $2(m-r)$.
 \item The GSB of  such a 
   $(n,2r+1,\delta)$ code is
          $(n-2r-1)(\floor{\frac{\delta}{2r+1)}} + \delta+1 =
          n-2r-1+\delta + 1 = n-2r+\delta = n-2r +2m-2r = 4(m-r)+1$.
    \item It may be shown that the code generated by $G[z]$ is an MDS
      convolutional MDS $(n,
        2r+1,2(m-r);1,4(m-r)+1)$ code. 
          
  \end{enumerate}
\end{Algorithm}

%% For each odd $2n+1$ there exist fields of characteristic $2$ over
%% which the Fourier $F_{2n+1}\times F_{2n+1}$ matrix may be
%% constructed. Find the order of $2 \mod 2n+1$, say  $2^s\equiv 1 \mod 2n+1$ and
%% then the field $GF(2^s)$ contains an element of order $2n+1$ with
%% which to construct the required Fourier matrix. %% Porbably the best
%% fileds to use are the $GF(2^{2^t})$ fields as these have precisely
%% elements of order $2^{2^s}-1$ elements and no `wastage'.

Consider the field $GF(2^n)$. This has an element of order $2^n-1=q$
and it seems best to construct the Fourier $F_q\times F_q$ over
$GF(2^n)$. % there is no `wastage'.

\begin{example} Construct $31\times 31$ DC
  convolutional codes.

{\em   The order of  $2 \mod 31$ is $5$ and thus work in the field $GF(2^5)$. Form the $F_{31}\times F_{31}$
matrix over $GF(2^5)$. Now proceed as in Algorithm \ref{algor1}. For
example let $\{A,B\}$ have rows \\ $\langle e_0, e_1, e_{30}, e_2, e_{29},e_3,e_{28},
e_4,e_{27},e_5,e_{26},
e_6,e_{25}, e_7,e_{24}, e_8,e_{23}\rangle,
\\ \langle 
0,0,0,e_9,e_{22},e_{10},e_{21},
e_{11},e_{20},e_{12},e_{19},e_{13},e_{18},e_{14},e_{17},e_{15},e_{16} \rangle$ respectively.

%% The $0$ in $B$ means the zero-vector of the right size.
Now form $G[z]=A+Bz$. The code generated by $G[z]$ is 
then  a $(31,17,14)$ DC convolutional code. The GSB for such a
code is $(n-r)(\floor{\delta}{r}+1) + \delta + 1 = (14)(1) +
14+1=29$. The generators of $A$ may be arranged in arithmetic sequence
with difference $1$ and so these form a $[31,17,15]$ MDS linear
code. Similarly the non-zero vectors in $B$ generate a $[31,14, 18]$
MDS linear code. Using these it may be shown that this code is an MDS
convolutional code. %% This amounts to showing  $(P_0+P_1z
%% + ...)(A+Bz)$ has distance $\geq 30$ except for $P(A+Bz)$ which has
%% distance $29$ when $P=(*,*,*, 0,\ldots, 0)$ which give a $(31,3,29)$
%% code or one with greater distance.

A quantum convolutional code may be designed from this.   
%Algebraic decoding methods exist. 

Larger rate DC convolutional MDS codes may also be
derived.

For example
let $A,B$ have rows \\ $\langle e_0, e_1, e_{30},e_2,e_{29},  e_3, e_{28},e_4,e_{27},e_5,e_{26},
e_6,e_{25}, e_7,e_{24}, e_8,e_{23}, e_9,e_{22},e_{10},e_{21}\rangle,
\\ \langle 
0,0,0, 0,0,0,0,
e_{11},e_{20},e_{12},e_{19},e_{13},e_{18},e_{14},e_{17},e_{15},e_{16} \rangle$ respectively.

This gives a $(31,21,10)$ DC code. The GSB for such a
code is $(n-r)(\floor{\frac{\delta}{r}}+1) + \delta + 1 = 10+11 =
21$. The free distance of this code is exactly $21$.
Similarly $(31,23,8)$ codes with free  distance $17$, $(31,25,6)$
with free distance $11$ and so on may be obtained.} 
\end{example}

\subsection{Addendum} It is shown in \cite{hurleyconv} how orthogonal
matrices may be used to construct convolutional codes. %% The
%% linear block matrices under this scheme gives LCD codes and these may
%% be `converted' to convolutional dual-containing codes under special
%% circumstances.
%% Methods for constructing suitable orthogonal matrices
%% are given in \cite{hurley44}.
Using orthogonal matrices does not allow the
same control on the distances achieved as can for Vandermonde/Fourier matrices.

Low density parity check (LDPC) codes have important applications in
communications. Linear block LDPC codes are constructed algebraically in \cite{hurley33} and
the methods can be extended to obtain convolutional LDPC codes. This
is dealt with separately.   

\nocite{*}
\end{document}